\documentclass[11pt,a4paper]{article}
\pdfoutput=1

\usepackage[margin=1.0in]{geometry}

\usepackage{url}
\usepackage{enumitem}
\usepackage[utf8]{inputenc}
\usepackage{amsmath,amssymb,euscript}
\usepackage{amscd} % simple commutative diagrams

\usepackage{graphicx}
\usepackage{color}

\usepackage{ntheorem}

\usepackage{float}

\floatstyle{ruled}
\newfloat{structure}{Hhtb}{los}
\floatname{structure}{Structure}

\usepackage{subfig}
\usepackage{algorithmic}
\usepackage{algorithm}
\usepackage{wrapfig}

\newcounter{prob}
        {\end{list}}
\newenvironment{proof}[1][{}]{%\novbskip%
  \begin{trivlist}\item[]\textit{Proof #1}\quad}%
  {\hfill\hspace*{\fill}~$\square$\end{trivlist}}

\theorembodyfont{\normalfont}
\newtheorem{thm}{Theorem} %[section]

\newtheorem{lem}[thm]{Lemma}
\newtheorem{de}[thm]{Definition}

\newtheorem{mainthm}{Theorem} % independently numbered
\newtheorem{hyps}{Hypotheses}

\theoremstyle{empty}
\newtheorem{dup}{Duplicate} % for duplicating a thm (give opt arg)

\definecolor{turquoise}{cmyk}{0.65,0,0.1,0.1}

\newcommand{\rawdef}[1]{\emph{#1}} % no index entry
\newcommand{\defn}[1]{\rawdef{#1}\index{#1}}

\newcommand{\Defref}[1]{Definition~\ref{#1}}
\newcommand{\Eqnref}[1]{Equation~\eqref{#1}}
\newcommand{\Figref}[1]{Figure~\ref{#1}}
\newcommand{\Lemref}[1]{Lemma~\ref{#1}}

\newcommand{\Secref}[1]{Section~\ref{#1}}

\newcommand{\Thmref}[1]{Theorem~\ref{#1}}

\DeclareMathOperator{\convh}{conv}
\newcommand{\convhull}[1]{\convh(#1)}
\newcommand{\bigo}[1]{O(#1)} %{\mathcal{O}(#1)}
\newcommand{\R}{\mathbb{R}}
\newcommand{\reel}{\mathbb{R}}

\newcommand{\rem}{\reel^m}

\newcommand{\norm}[1]{\lVert#1\rVert}
\newcommand{\abs}[1]{\lvert#1\rvert}
\newcommand{\transp}[1]{{#1}^\mathsf{T}}%\intercal}%\mathsf{T}}

\newcommand{\card}[1]{\#{#1}} %{\sharp{#1}}

\newcommand{\inv}[1]{{#1}^{-1}}
\newcommand{\invtransp}[1]{{#1}^{-T}}

\newcommand{\bdry}[1]{\partial{#1}}

\DeclareMathOperator{\starr}{star}
\newcommand{\str}[1]{\starr(#1)}

\newcommand{\asimplex}[1]{\{#1\}} % abstract simplex, i.e. $\asimplex{i,j,k}$
\newcommand{\simplex}[1]{[#1]} % Euclidean simplex
\newcommand{\carrier}[1]{\abs{#1}} % carrier of simplical complex

\newcommand{\ambdim}{N}

\newcommand{\amb}{\reel^{\ambdim}} % ambient space

\newcommand{\gdist}{d} % generic metric (distance function)
\newcommand{\dist}[2]{\gdist(#1,#2)}
\newcommand{\gdistG}[1]{\gdist_{#1}} % metric on space to be specified
\newcommand{\distG}[3]{\gdistG{#1}(#2,#3)}
\newcommand{\distEm}[2]{\distG{\rem}{#1}{#2}}

\newcommand{\gdistM}{\gdistG{\man}}
\newcommand{\distM}[2]{\distG{\man}{#1}{#2}}

\newcommand{\close}[1]{\overline{#1}} % topological closure
 
\newcommand{\ball}[2]{B(#1,#2)} % generic ball
\newcommand{\cball}[2]{\close{B}(#1,#2)} % generic closed ball
\newcommand{\spaceball}[3]{B_{#1}(#2,#3)} % open ball on specified space
\newcommand{\ballEm}[2]{\spaceball{\rem}{#1}{#2}} % open ball in R^m
\newcommand{\cballEm}[2]{\cspaceball{\rem}{#1}{#2}} % closed ball in R^m

\newcommand{\ballM}[2]{\spaceball{\man}{#1}{#2}}

\DeclareMathOperator{\aff}{aff} % affine hull
\newcommand{\affhull}[1]{\aff(#1)}

\newcommand{\pts}{\mathsf{P}}

\newcommand{\man}{\mathcal{M}}
\newcommand{\tman}{\tilde{\man}}

\DeclareMathOperator{\Del}{Del}
\newcommand{\del}[2]{\Del_{#1}(#2)} % #1 metric identifier; #2 points
\newcommand{\delof}[1]{\Del(#1)}
\newcommand{\delP}{\delof{\pts}}

\newcommand{\pertconst}{\rho}

\newcommand{\samconst}{\epsilon}
\newcommand{\tsamconst}{\tilde{\epsilon}}
\newcommand{\sparseconst}{\mu_0} %min d(p,q) > \sparseconst\samconst
\newcommand{\pert}{\zeta} % perturbation function: P \to \tilde{P}

\newcommand{\sing}[2]{s_{#1}(#2)}
\newcommand{\pseudoinv}[1]{#1^\dagger}

\newcommand{\splxs}{\sigma}
\newcommand{\tsplxs}{\tilde{\sigma}}
\newcommand{\splxt}{\tau}
\newcommand{\tsplxt}{\tilde{\tau}}

\newcommand{\splxjoin}[2]{{#1}*{#2}}
\newcommand{\normhull}[1]{N(#1)}
\newcommand{\opface}[2]{#2_{#1}} % e.g., \sigma_p : face opposite p

\newcommand{\splxtp}{\opface{p}{\splxt}}

\newcommand{\thickbnd}{\Upsilon_0}
\newcommand{\flakebnd}{\Gamma_0}
\newcommand{\tflakebnd}{\tilde{\Gamma}_0}
\newcommand{\thickness}[1]{\Upsilon(#1)}

\newcommand{\splxalt}[2]{D(#1,#2)} % altitude of #1 above face #2
\newcommand{\longedge}[1]{\Delta(#1)}
\newcommand{\shortedge}[1]{L(#1)}
\newcommand{\circrad}[1]{R(#1)}
\newcommand{\circcentre}[1]{C(#1)}
\DeclareMathOperator{\injr}{inj}
\newcommand{\injrad}[1]{\injr(#1)}
\newcommand{\injradM}{\injrad{\man}}

\renewcommand{\distEm}[2]{\dist{#1}{#2}}
\renewcommand{\ballEm}[2]{\ball{#1}{#2}} % open ball in R^m
\renewcommand{\cballEm}[2]{\cball{#1}{#2}} % closed ball in R^m

\newcommand{\pertbnd}{\rho_0} %{\eta_0}
\newcommand{\tpertbnd}{\tilde{\rho}_0} %{\eta_0}
\newcommand{\psparseconst}{\sparseconst'}
\newcommand{\psamconst}{\samconst'}
\newcommand{\ppts}{\pts'}

\newcommand{\rdelsmhull}[1]{\Del_{|}(#1)}

\newcommand{\mueps}{(\sparseconst, \samconst)}
\newcommand{\pmueps}{(\psparseconst, \psamconst)}
\newcommand{\ueset}{$\mueps$-net}
\newcommand{\pueset}{$\pmueps$-net}

\newcommand{\dgconfig}{forbidden configuration} %{$\dg$-configuration}
\newcommand{\hoopbnd}{\alpha_0}
\newcommand{\thoopbnd}{\tilde{\alpha}_0}
\newcommand{\hoop}{$\hoopbnd$-hoop}

\newcommand{\circsphere}[1]{S(#1)}
\newcommand{\diasphere}[1]{S^{m-1}(#1)}

\newcommand{\apts}{\mathcal{N}}
\newcommand{\nbrs}[1]{\apts_{#1}}
\newcommand{\nbrsi}{\nbrs{i}}

\newcommand{\chart}[1]{\phi_{#1}}
\newcommand{\charti}{\chart{i}}
\newcommand{\pchart}[1]{\phi'_{#1}}
\newcommand{\pcharti}{\pchart{i}}

\newcommand{\fchart}[1]{\varphi_{#1}} % full homeomorphism chart
\newcommand{\fcharti}{\fchart{i}}
\newcommand{\tchart}[1]{\tilde{\varphi}_{#1}}
\newcommand{\tcharti}{\tchart{i}}

\newcommand{\trans}[2]{\varphi_{#2 #1}} % from #1 to #2
\newcommand{\transij}{\trans{i}{j}}
\newcommand{\outmesh}{\delof{\ppts}}

\newcommand{\gdisti}{\gdistG{i}}
\newcommand{\disti}[2]{\distG{i}{#1}{#2}}
\newcommand{\gdistj}{\gdistG{j}}
\newcommand{\distj}[2]{\distG{j}{#1}{#2}}
\newcommand{\gdistk}{\gdistG{k}}
\newcommand{\distk}[2]{\distG{k}{#1}{#2}}

\newcommand{\balli}[2]{\spaceball{i}{#1}{#2}} % open ball in R^m
\newcommand{\ballj}[2]{\spaceball{j}{#1}{#2}} % open ball in R^m
\newcommand{\metpert}{\xi}

\newcommand{\metlipconst}{\xi_0}

\newcommand{\lipsamconst}{\samconst_0}

\newcommand{\qpts}{\mathsf{S}}

\newcommand{\lpts}{\mathsf{Q}'}
\newcommand{\lptsi}{\mathsf{Q}'_i}
\newcommand{\lptsj}{\mathsf{Q}'_j}

\usepackage[pagebackref=true]{hyperref}

\title{Delaunay triangulation of manifolds
}

\author{
Jean-Daniel Boissonnat
\footnote{
This research has been partially supported by the 7th Framework
Programme for Research of the European Commission, under FET-Open
grant number 255827 (CGL Computational Geometry Learning). Partial 
support has also been provided by the Advanced Grant of the European 
Research Council GUDHI (Geometric Understanding in Higher Dimensions).
}
\\
{\normalsize Geometrica}\\
{\normalsize INRIA}\\
{\normalsize Sophia-Antipolis, France}\\
\url{Jean-Daniel.Boissonnat@inria.fr}
\and
Ramsay Dyer
\footnotemark[1]
\\
{\normalsize Johann Bernoulli Institute}\\
{\normalsize University of Groningen}\\ 
{\normalsize Groningen, The Netherlands}\\
\url{r.h.dyer@rug.nl}
\and
Arijit Ghosh
\footnote{
Supported by the Indo-German Max Planck Center for Computer Science (IMPECS).
}
\footnotemark[1]
\footnote{
Part of the work was done when the author was a visiting scientist at 
ACM Unit, Indian Statistical Institute, Kolkata, India.
}
\\
{\normalsize D1: Algorithms \& Complexity}\\
{\normalsize Max-Planck-Institut f\"ur Informatik}\\
{\normalsize Saarbr\"ucken, Germany}\\
\url{agosh@mpi-inf.mpg.de}
}

\begin{document}

\pagenumbering{roman}
\maketitle
% -*- LaTeX -*-
% abstract.tex
% 20130527
%
% for man_mesh

\begin{abstract}
  We present an algorithm for producing Delaunay triangulations of manifolds.  The algorithm can accommodate abstract manifolds that are not presented as submanifolds of Euclidean space. Given a set of sample points and an atlas on a compact manifold, a manifold Delaunay complex is produced provided the transition functions are bi-Lipschitz with a constant close to 1, and the sample points meet a local density requirement; no smoothness assumptions are required. If the transition functions are smooth, the output is a triangulation of the manifold.

  The output complex is naturally endowed with a piecewise flat metric which, when the original manifold is Riemannian, is a close approximation of the original Riemannian metric.  In this case the ouput complex is also a Delaunay triangulation of its vertices with respect to this piecewise flat metric.

  
\end{abstract}

%\paragraph{Keywords.} Delaunay complex, bi-Lipschitz function, manifold, protection, perturbation

\thispagestyle{empty}

\clearpage
\tableofcontents

\clearpage

\pagenumbering{arabic}

% -*- LaTeX -*-
% intro.tex
% 20140408
% Another crack at an intro for the man-mesh paper (this time initially for the MFCS conference submission) 

\section{Introduction}
\label{sec:intro}

We present an algorithm for computing Delaunay triangulations of Riemannian manifolds. Not only is this the first algorithm guaranteed to produce a Delaunay triangulation of an arbitrary compact Riemannian manifold, it also provides the first theoretical demonstration of the existence of such triangulations on manifolds of dimension greater than 2 with non-constant curvature. 

The Delaunay complex is a natural structure to consider when seeking to triangulate a space equipped with a metric. It plays a central role in the development of algorithms for meshing Euclidean domains. In applications where an anisotropic mesh is desired, a standard approach is to consider a Riemannian metric defined over the domain and to construct an approximate Delaunay triangulation with respect to this metric~\cite{labelle2003,boissonnat2011aniso.tr,canas2012}. In this context we can consider the domain to be a Riemannian manifold that admits a global coordinate parameterisation. The algorithm we present here encompasses this setting, modulo boundary considerations.

In the case of surfaces, it has been shown that a Delaunay triangulation exists if the set of sample points is sufficiently dense \cite{leibon1999,dyer2008sgp}. However, in higher dimensional manifolds, contrary to previous claims~\cite{leibon2000}, sample density alone is not sufficient to ensure a triangulation \cite[App. A]{boissonnat2013stab2.inria}. 

In Euclidean space, $\rem$, the Delaunay complex on a set of points $\pts$ is a triangulation provided the points are \defn{generic}, i.e., no ball empty of points contains more than $m+1$ points of $\pts$ on its boundary \cite{delaunay1934}. A point set that is not generic is said to be \defn{degenerate}, and such configurations can be avoided with an arbitrarily small perturbation. However, when the metric is no longer homogeneous, an arbitrarily small perturbation is not sufficient to guarantee a triangulation. In previous work~\cite{boissonnat2013stab1} we have shown that genericity can be parameterised, with the parameter, $\delta$, indicating how far the point set is from degeneracy. A $\delta$-generic point set in Euclidean space yields a Delaunay triangulation that is quantifiably stable with respect to perturbations of the metric, or of the point positions. We later produced an algorithm~\cite{boissonnat2014flatpert} that, given an initial point set $\pts \subset \rem$, generates a perturbed point set $\pts'$ that is $\delta$-generic. The algorithm we present here adapts this Euclidean perturbation algorithm to the context of compact manifolds equipped with a metric that can be locally approximated by a Euclidean metric. In particular, this includes Riemannian metrics, as well as the extrinsic metric on submanifolds, which defines the so-called restricted Delaunay complex, variations of which have been exploited in algorithms for reconstructing submanifolds of Euclidean space from a finite set of sample points~\cite{cheng2005,boissonnat2014tancplx.dcg}.

The simplicial complex produced by our algorithm is naturally equipped with a piecewise flat metric that is a quantifiably good approximation to the metric on the original manifold. The stability properties of the constructed Delaunay triangulation yield additional benefits. In particular, the produced complex is a Delaunay triangulation of its vertices with respect to its own intrinsic piecewise-flat metric: a self-Delaunay complex. Such complexes are of interest in discrete differential geometry because they provide a natural setting for discrete exterior calculus~\cite{bobenko2007,dyer2010thesis,hirani2012}.

% obligatory? contents in a paragraph:
In \Secref{sec:pert.strategy} we review the main ideas involved in the perturbation algorithm~\cite{boissonnat2014flatpert} for producing $\delta$-generic point sets in Euclidean space. The extension of the algorithm to general manifolds is described in \Secref{sec:overview}, where we also state our main results. Sections~\ref{sec:analysis.outline} and \ref{sec:details} describe the analysis that leads to these results.

\paragraph{Contributions}
In this paper we provide the first proof of existence of Delaunay triangulations of arbitrary compact Riemannian manifolds. The proposed algorithm is the first triangulation algorithm that can accommodate abstract manifolds that are not presented as submanifolds of Euclidean space. 
The output complex is a good geometric approximation to the original manifold, and also posesses the self-Delaunay property. These results are summarised in \Thmref{thm:riem.del.tri}.

The algorithm accommodates more general inputs than Riemannian manifolds: strong bi-Lipschitz constraints are required on the transition functions, but they need not be smooth. In this case we cannot guarantee a triangulation, but, as stated in Theorems~\ref{thm:man.mesh} and \ref{thm:output.del.cplx}, the output complex is a manifold Delaunay complex (it is the nerve of the Voronoi diagram of the perturbed points $\ppts \subset \man$, and it is a manifold).

The framework encompasses and unifies previous algorithms for constructing anisotropic meshes, and for meshing submanifolds of Euclidean space, and the algorithm itself is conceptually simple (if not the analysis).

% -*- LaTeX -*-
% strategy.tex
% 20140410
% Outline of the perturbation strategy for the conference paper; introduce forbidden configurations, notation, etc., but without the detail we have in our review of flat-pert in the full version.

\section{The perturbation strategy}
\label{sec:pert.strategy}

We outline here the main ideas behind the Euclidean perturbation algorithm~\cite{boissonnat2014flatpert}, upon which the current algorithm is based.  Given a set $\pts \subset \rem$, that algorithm produces a perturbed point set $\pts'$ that is $\delta$-generic. This means that the Delaunay triangulation of $\pts'$ will not change if the metric is distorted by a small amount. 

% There is a more detailed summary in the full version of this paper~\cite{boissonnat2013manmesh.inria}.

%\subsection{Thickness and protection}
\paragraph{Thickness and protection}
We consider a finite set $\pts \subset \rem$.  A simplex $\splxs \subset \pts$ is a finite collection of points: $\splxs = \asimplex{p_0, \ldots, p_j}$. We work with abstract simplices, and in particular $x \in \splxs$ means $x$ is a vertex of $\splxs$.  The join of two simplices, $\splxjoin{\splxt}{\splxs}$, is the union of their vertices.  By the standard abuse of notation, a point $p$ may represent the $0$-simplex $\{p\}$.  Although we prefer abstract simplices, we freely talk about standard geometric properties, such as the longest edgelength, $\longedge{\splxs}$, and the length of the shortest edge $\shortedge{\splxs}$.

For $p \in \splxs$, $\opface{p}{\splxs}$ is the facet opposite $p$, and $\splxalt{p}{\splxs}$ is the \defn{altitude} of $p$ in $\splxs$, i.e., $\splxalt{p}{\splxs}= \distEm{p}{\affhull{\opface{p}{\splxs}}}$.  The \defn{thickness} of $\splxs$ is a measure of the quality of $\splxs$, and is denoted $\thickness{\splxs}$. If $\splxs$ is a $0$-simplex, then $\thickness{\splxs}=1$. Otherwise $\thickness{\splxs}$ is the smallest altitude of $\splxs$ divided by $j\longedge{\splxs}$, where $j$ is the dimension of $\splxs$.  If $\thickness{\splxs} = 0$, then $\splxs$ is \defn{degenerate}.  We say that $\splxs$ is $\thickbnd$-thick, if $\thickness{\splxs} \geq \thickbnd$. If $\splxs$ is $\thickbnd$-thick, then so are all of its faces.

A \defn{circumscribing ball} for a simplex $\splxs$ is any $m$-dimensional ball that contains the vertices of $\splxs$ on its boundary.  A degenerate simplex may not admit any circumscribing ball.  If $\splxs$ admits a circumscribing ball, then it has a \defn{circumcentre}, $\circcentre{\splxs}$, which is the centre of the unique smallest circumscribing ball for $\splxs$. The radius of this ball is the \defn{circumradius} of $\splxs$, denoted $\circrad{\splxs}$.  

A ball $\ballEm{c}{r}$ is open, and $\cballEm{c}{r}$ is its closure.
The Delaunay complex, $\delP$ is the (abstract) simplicial complex
defined by the criterion that a simplex belongs to $\delP$ if it has a
circumscribing ball whose intersection with $\pts$ is empty. For $p \in \pts$, the \defn{star} of $p$ is the subcomplex $\str{p;\delP}$ consisting of all simplices that contain $p$, as well as the faces of those simplices. An
$m$-simplex $\splxs^m \in \delP$ is \defn{$\delta$-protected} if
$\cballEm{\circcentre{\splxs^m}}{\circrad{\splxs^m} + \delta} \cap
\pts = \splxs^m$. The point set $\pts \subset \rem$ is \defn{$\delta$-generic} if all the $m$-simplices in $\delP$ are $\delta$-protected.

%\subsection{Forbidden configurations}
\paragraph{Forbidden configurations}
The essential observation that leads to the perturbation algorithm is that if $\pts \subset \rem$ is such that there exists a Delaunay $m$-simplex that is not $\delta$-protected, then there is a \defn{forbidden configuration}: a (possibly degenerate) simplex $\splxt \subset \pts$ characterised by the properties we describe in \Lemref{lem:prop.forbid.cfg} below.
%that has the property that every vertex is close to the circumsphere of the opposing facet.  
We emphasise that a forbidden configuration need not be a Delaunay simplex.  The perturbation algorithm guarantees that the Delaunay $m$-simplices will be $\delta$-protected by ensuring that each point is perturbed to a position that is not too close to the circumsphere of any of the nearby simplices in the current (perturbed) point set. A volumetric argument shows that this can be achieved.

%In \Lemref{lem:prop.forbid.cfg} below, we use parameters describing the point set to describe forbidden configurations.  
If $D \subset \rem$, then $\pts$ is \defn{$\samconst$-dense} for $D$ if $\distEm{x}{\pts} < \samconst$ for all $x \in D$.  We refer to $\samconst$ as the \defn{sampling radius}.  The set $\pts$ is \defn{$\sparseconst \samconst$-separated} if $\distEm{p}{q} \geq \sparseconst \samconst$ for all $p,q \in \pts$, and $\pts$ is a \ueset\ (for $D$) if it is $\sparseconst \samconst$-separated, and $\samconst$-dense (for $D$).

The goal is to produce a perturbed point set $\ppts$ that contains no forbidden configurations.  A \defn{$\pertconst$-perturbation} of a \ueset\ $\pts \subset \rem$ is a bijective application $\pert: \pts \to \ppts \subset \rem$ such that $\distEm{\pert(p)}{p} \leq \pertconst$ for all $p \in \pts$.  Unless otherwise specified, a \defn{perturbation} will always refer to a $\pertconst$-perturbation, with $\pertconst = \pertbnd \samconst$ for some $\pertbnd \leq \frac{\sparseconst}{4}$.  We also refer to $\ppts$ itself as a perturbation of $\pts$.  We generally use $p'$ to denote the point $\pert(p) \in \ppts$, and similarly, for any point $q' \in \ppts$ we understand $q$ to be its preimage in $\pts$.  We observe \cite[Lemma 2.2]{boissonnat2014flatpert} that $\ppts$ is a \pueset, with $\psamconst \leq \frac{5}{4}\samconst$ and $\psparseconst \geq \frac{2}{5}\sparseconst$.

% A forbidden configuration is a specific kind of poorly shaped simplex that has the property that \emph{all} its altitudes are small. 
Given a positive parameter $\flakebnd \leq 1$, we say that $\splxs$ is \defn{$\flakebnd$-good} if for all $\splxs^j \subseteq \splxs$, we have $\thickness{\splxs^j} \geq \flakebnd^j$, where $\flakebnd^j$ is the $j^{\text{th}}$ power of $\flakebnd$.
% A simplex that is $\flakebnd$-good is necessarily $\flakebnd^m$-thick, but the converse is not generally true. 
A \defn{$\flakebnd$-flake} is a simplex that is not $\flakebnd$-good, but whose facets all are. A flake may be degenerate.  The altitudes of a flake are subjected to an upper bound proportional to $\flakebnd$.

If a simplex is not $\flakebnd$-good, then it necessarily contains a
face that is a flake. This follows easily from the observation that
$\thickness{\splxs} = 1$ if $\splxs$ is a $1$-simplex. If $\splxs^m
\in \delP$ is not $\delta$-protected, then there is a $q \in \pts
\setminus \splxs^m$ that is within a distance $\delta$ of the
circumsphere of $\splxs^m$. Since $\splxjoin{q}{\splxs^m}$ is
$(m+1)$-dimensional, it is degenerate, and therefore has a face
$\splxt$ that is a $\flakebnd$-flake. Such a $\splxt$ is a forbidden
configuration.

If a simplex $\splxs$ has a circumcentre, we define the \defn{diametric sphere} as the boundary of the smallest circumscribing ball: $\diasphere{\splxs} = \bdry{\ballEm{\circcentre{\splxs}}{\circrad{\splxs}}}$, and the \defn{circumsphere}: $\circsphere{\splxs} = \diasphere{\splxs} \cap \affhull{\splxs}$. If $\splxs \subset \splxt$, then $\circsphere{\splxs} \subseteq \circsphere{\splxt}$, and if $\dim \splxs = m$, then $\circsphere{\splxs} = \diasphere{\splxs}$.

The bound on the altitudes, together with the stability property of circumscribing balls of thick simplices, allows us to demonstrate that forbidden configurations have the \defn{$\hoopbnd$-hoop property}.
A $k$-simplex $\splxt$ has the $\hoopbnd$-hoop property if for every $(k-1)$-facet $\splxs \subset \splxt$ we have
\begin{equation*}
  \distEm{p}{\circsphere{\splxs}} \leq \hoopbnd
  \circrad{\splxs} < \infty,
\end{equation*}
where $p$ is the vertex of $\splxt$ not in $\splxs$.

We are concerned with forbidden configurations in the perturbed point set $\ppts$. In addition to the two parameters that describe a \pueset, forbidden configurations depend on the flake parameter $\flakebnd$, as well as the parameter $\delta_0$, which governs the protection via the requirement $\delta = \delta_0 \psparseconst \psamconst$. 
% Introduced this defn because it is hard to make sense of
% Lemma~\ref{lem:loc.eucl.protect} without it
\begin{de}[Forbidden configuration]
  \label{def:forbidden.config}
  Let $\pts' \subset \rem$ be a \pueset.  A $(k+1)$-simplex $\splxt
  \subseteq \pts'$, is a \defn{\dgconfig} in $\ppts$ if it is a
  $\flakebnd$-flake, with $k\leq m$, and there exists a $p \in \splxt$
  such that $\opface{p}{\splxt}$ has a circumscribing ball $B =
  \ballEm{C}{R}$ with $R < \psamconst$, and $\abs{\distEm{p}{C} - R}
  \leq \delta$, where $\delta = \delta_0 \psparseconst \psamconst$. 
\end{de}
The definition itself is awkward, but for most purposes we can simply refer to the following summary \cite[Theorem 3.10]{boissonnat2014flatpert} of properties of forbidden configurations in $\ppts$ in terms of the parameters of the original point set $\pts$:
\begin{lem}[Properties of forbidden configurations]
  \label{lem:prop.forbid.cfg} % for the properties
  Suppose that $\pts \subset \rem$ is a \ueset\ and that $\ppts$ is a
  $\pertbnd\samconst$-perturbation of $\pts$, with $\pertbnd \leq \frac{\sparseconst}{4}$.
  If
  \begin{equation}
    \label{eq:dnobnd}
    \delta_0 \leq \flakebnd^{m+1} \quad \text{and} \quad \flakebnd
    \leq \frac{2\sparseconst^2 }{75},
  \end{equation}
  then every forbidden configuration $\splxt \subset \ppts$
  satisfies \emph{all} of the following properties:
  \begin{enumerate}[label={$\mathcal{P}$\arabic*}]
  \item \label{hyp:clean.hoop.bnd}
    Simplex $\splxt$ has the \hoop\ property, with
    $\hoopbnd = \frac{ 2^{13} \flakebnd }{\sparseconst^3}$.
  \item \label{hyp:clean.facet.rad.bnd}
    For all $p \in \splxt$,
    $\circrad{\splxtp} < 2\samconst$.
  \item \label{hyp:diam.bnd}
   $\longedge{\splxt} < \frac{5}{2}(1 +
   \frac{1}{2}\delta_0\sparseconst)\samconst$.
  \item \label{hyp:good.facets}
    Every facet of $\splxt$ is $\flakebnd$-good.
  \end{enumerate}
\end{lem}

The algorithm focuses on Property~\ref{hyp:clean.hoop.bnd} of
forbidden configurations.  A critical aspect of this property is its
symmetric nature; if we can ensure that $\splxt$ has one vertex that is
not too close to its opposite facet, then $\splxt$ cannot be a
forbidden configuration.

Using Property~\ref{hyp:diam.bnd}, we can find, for each $p \in
\pts$, a complex $\mathcal{S}_p$ consisting of all simplices $\splxs \in
\pts$ such that after perturbations $\splxjoin{p}{\splxs}$ could be a
forbidden configuration.

The algorithm proceeds by perturbing each point $p \in \pts$ in turn, such that each point is only visited once. The perturbation $p \mapsto p'$ is found by randomly trying perturbations $p \mapsto x$ until it is found that $x$ is a \defn{good perturbation}. A good perturbation is one in which $\distEm{x}{\circsphere{\splxs}} > 2\hoopbnd \samconst$ for all $\splxs \in \mathcal{S}_p(\ppts)$, where $\mathcal{S}_p(\ppts)$ is the complex in the current perturbed point set whose simplices correspond to those in $\mathcal{S}_p$. By Property~\ref{hyp:clean.hoop.bnd}, $\splxjoin{x}{\splxs}$ cannot be a forbidden configuration.

A volumetric argument based on the finite number of simplices in
$\mathcal{S}_p$, the small size of $\hoopbnd$, and the volume of the
ball $\ballEm{p}{\pertbnd \samconst}$ of possible perturbations of
$p$, reveals a high probability that $p \mapsto x$ will be a good
perturbation, and thus ensures that the algorithm will terminate.

Upon termination there will be no forbidden configurations in $\ppts$,
because every perturbation $p \mapsto p'$ ensures that there are no
forbidden configurations incident to $p'$ in the current point set,
and no new forbidden configurations are introduced.

% -*- LaTeX -*-
% overview.tex
% 20130605
% Overview for the man-mesh paper.

%\section{Overview of the extended algorithm}
\section{Overview and main results}
\label{sec:overview}

The extension of the perturbation algorithm to the curved setting is
accomplished by performing the perturbations, and the analysis, in
local Euclidean coordinate patches. The main idea is that forbidden configurations exhibit some stability with respect to small changes in the Euclidean metric. In particular, if we ensure there are no forbidden configurations in some region of one Euclidean coordinate patch, then, assuming a slightly smaller hoop parameter $\hoopbnd$, there will be no forbidden configurations in the corresponding region of any nearby coordinate patche. This means that the perturbed point set will be $\delta$-generic in any local Euclidean coordinate patch, and the resulting stability of the local Euclidean Delaunay triangulations ensures that they will agree on neighbouring patches.

We assume we have a finite set of points $\pts$ in a compact manifold $\man$. It is convenient to employ an index set $\apts$ of unique (integer) labels for $\pts$, thus we employ a bijection $\iota: \apts \to \pts \subset \man$. We assume that $\pts$ is sufficiently dense that we may define an atlas $\{(W_i, \fcharti)\}_{i \in \apts}$ for $\man$ such that the coordinate charts $\fcharti: W_i \to U_i \subset \rem$ have low metric distortion, as defined in \Secref{sec:input}. The set $W_i$ is required to contain a sufficiently large ball centred at the point indexed by $i$.  We refer to $U_i$ as a \defn{coordinate patch}.

We work exclusively in the Euclidean coordinate patches $U_i$, exploiting the transition functions $\transij = \fchart{j} \circ \fcharti^{-1}$ to translate between them.  We define $\pts_i = \fcharti(W_i \cap \pts)$, but given these sets, the algorithm itself makes no explicit reference to either $\pts$ or to the coordinate charts $\fcharti$, except to keep track of the labels of the points. We employ the discrete map $\charti = \fcharti \circ \iota$ to index the elements of the set $\pts_i$.

The idea is to  perturb $p_i = \charti(i) \in U_i$ in such a way, $p_i \mapsto p'_i$, that not only are there no forbidden configurations incident to $p'_i$ in $\ppts_i = \fcharti(W_i \cap \ppts)$, but there are no forbidden configurations incident to $\transij(p_i') \in \ppts_j \subset U_j$ either, where $j$ is the index of any sample point near $p_i$.

Before detailing the requirements of the input data in \Secref{sec:input}, we briefly discuss the implicit and explicit properties of the underlying manifold $\man$ in \Secref{sec:in.man.properties}. A summary of the analysis and main results is presented in \Secref{sec:analysis.outline}.
%The quality of the output complex is discussed in \Secref{sec:quality}.

\subsection{Manifolds represented by transition functions}
\label{sec:in.man.properties}

The essential input data for the algorithm are the transition functions,
and the sample points in the coordinate patches; we do not explicitly
use the coordinate charts or the metric on the manifold. However, given that the transition functions can be defined by an atlas on a manifold, this manifold is essentially unique:
If $\tman$ has an atlas $\{(\tilde{W}_i, \tcharti)\}_{i \in \apts}$ such that $\tcharti(\tilde{W}_i) = U_i$ and $\transij = \tchart{j} \circ \tcharti^{-1}$ for all $i,j \in \apts$, then $\tman$ and $\man$ are homeomorphic. Indeed, we define the map $f: \man \to \tman$ by $f(x) = \tcharti^{-1} \circ \fcharti(x)$ if $x \in W_i$. The map is well defined, because $\fchart{j} = \transij \circ \fcharti$ on $W_i \cap W_j$, and $\tchart{j}^{-1} = \tcharti^{-1} \circ \transij^{-1}$ on $U_{ji} = \fchart{j}(W_i \cap W_j)$. It can be verified directly from the definition that $f$ is a homeomorphism, since it is bijective and locally a homeomorphism.

Although the algorithm does not make explicit reference to a metric on the manifold $\man$, the metric distortion bounds required on the transition functions imply a metric constraint. Implicitly we are using a metric on the manifold for which the coordinate charts have low metric distortion. 
If a metric on the manifold is not explicitly given then, at least in the case where the transition functions are smooth, we can be sure that such a metric exists: Given the coordinate charts an appropriate Riemannian metric on the manifold can be obtained from the coordinate patches by the standard construction employing a partition of unity subordinate to the atlas (e.g., \cite[Thm. V.4.5]{boothby1986}).

Thus, although the manifold may be presented abstractly in terms of coordinate patches and transition functions between them, this information essentially characterises the manifold. The algorithm we present is not a reconstruction algorithm, it is an algorithm to triangulate a known manifold.

\subsection{The setting and input data}
\label{sec:input}

We take as input a finite index set $\apts = \{1, \ldots, n \}$, which
we might think of as an abstract set of points (without geometry),
together with the geometric data we will now introduce. The details of
the arguments that lead to our choices in the size of the domains are
given in \Secref{sec:domain.size}.

\paragraph{Coordinate patches.}
For each $i \in \apts$ we have a neighbourhood set $\nbrsi \subseteq
\apts$, a sampling radius $\samconst_i > 0$,  and an injective
application
\begin{equation*}
  \charti: \nbrsi \to U_i \subseteq \rem,
\end{equation*}
such that $\pts_i = \charti(\nbrsi)$ is a $(\sparseconst,\samconst_i)$-net for $\ballEm{p_i}{8\samconst_i} \subseteq U_i$, where we adopt the notation $p_j = \charti(j)$ for any $j \in \nbrsi$. The separation parameter $\sparseconst$ is globally defined to be the same on all coordinate patches, but the sampling radius $\samconst_i$ may be different in different patches, subject to a mild constraint described below. We call the standard metric on $U_i \subseteq \rem$ the \defn{local Euclidean metric} for $i$, and we will denote it by $\gdisti$ to distinguish between the different local Euclidean metrics. Similarly, $\balli{c}{r}$ denotes a ball with respect to the metric $\gdisti$.

\paragraph{Transition functions.}
For each $p_j \in \balli{p_i}{6\samconst_i} \subset U_i$ we require a
neighbourhood $U_{ij} \subseteq U_i$ such that
\begin{equation*}
  \balli{p_i}{6\samconst_i} \cap \balli{p_j}{9\samconst_i} \subseteq U_{ij}
\end{equation*}
and 
\begin{equation*}
  U_{ij} \cap \charti(\nbrsi) = \charti(\nbrsi \cap \nbrs{j}).
\end{equation*}
The set $U_{ij}$ is the domain of the \defn{transition function}
$\transij$, which is a homeomorphism
\begin{equation*}
  \transij : U_{ij} \to U_{ji} \subseteq U_j,
\end{equation*}
such that $\transij = \trans{j}{i}^{-1}$ and
\begin{equation*}
  \transij \circ \charti = \chart{j} \quad \text{ on } \quad \nbrsi
  \cap \nbrs{j}.
\end{equation*}
These transition functions are required to have \emph{low metric
  distortion}:
\begin{equation}
  \label{eq:metric.distortion}
  \abs{\disti{x}{y} - \distj{\transij(x)}{\transij(y)}} \leq 
  \metlipconst \disti{x}{y}\quad \text{for all } x,y \in U_{ij},
\end{equation}
where $\metlipconst > 0$ is a small positive parameter that quantifies the metric distortion. We say that $\transij$ is a \defn{$\metlipconst$-distortion map}. 

In order to ease the notational burden, $\charti(j) \in U_{ik}$, and
$\chart{k}(j) \in U_{ki}$, are denoted by the same symbol,
$p_j$. Ambiguities are avoided by distinguishing between the Euclidean
metrics $\gdisti$ and $\gdistk$. Although $\gdisti$ is the canonical
metric on $U_{ik}$, we may consider the pullback of $\gdistk$ from the
homeomorphic domain $U_{ki}$. Thus for $x,y \in U_{ik}$ the expression
$\distk{x}{y}$ is understood to mean
$\distk{\trans{i}{k}(x)}{\trans{i}{k}(y)}$, but we also occasionally
employ the latter, redundant, notation.

Using symmetry, we observe that \Eqnref{eq:metric.distortion} implies that
\begin{equation*}
  \abs{\disti{x}{y} - \distj{x}{y}} \leq 
  \metlipconst \min \{ \disti{x}{y}, \distj{x}{y} \}.
\end{equation*}
Our analysis will require that $\metlipconst$ be very small. For
standard coordinate charts, $\metlipconst$ can be shown to be
$\bigo{\samconst}$, where $\samconst$ is a sampling radius on the
manifold. For example, this is the case when considering a smooth
submanifold of $\amb$, and using the orthogonal projection onto the
tangent space as a coordinate chart \cite[Lemma
3.7]{boissonnat2013stab2.inria}. Thus $\metlipconst$ may be made as
small as desired by increasing the sampling density. 

\paragraph{Adaptive sampling.}
We will further require a constraint on the difference between
neighbouring sampling radii:
\begin{equation*}
  \abs{\samconst_i - \samconst_j} \leq \lipsamconst
  \min\{\samconst_i,\samconst_j\},
\end{equation*}
whenever $\disti{p_i}{p_j} \leq 6 \samconst_i$.  This allows us to
work with a constant sampling radius in each coordinate frame, while
accommodating a globally adaptive sampling radius.

For example, suppose $\samconst: \man \to \reel$ is a positive, $\nu$-Lipschitz function, with respect to the metric $\gdistM$ on the manifold. Then $\samconst$ may be used as an adaptive sampling radius on $\man$, i.e., $\pts \subset \man$ is $\samconst$-dense if $\distM{x}{\pts} < \samconst(x)$ for all $x \in \man$. A popular example of such a function is $\samconst(x) = \nu f(x)$, where $f$ is the ($1$-Lipschitz) \defn{local feature size}~\cite{amenta1999vf}.

% The following considerations are worked out in CF13p43
% Actually the first calc is easy, but I keep confusing it: see postage
% stamp calculation near bottom of CF13p39
Using the $\nu$-Lipschitz continuity of $\samconst$, we can define, for any $p_i \in \man$, a constant $\tsamconst_i$, such that $\pts$ is $\tsamconst_i$-dense in some neighbourhood of $p_i$. In fact, given $c > 0$, with $c < \nu^{-1}$, we find that $\pts$ is $\tsamconst_i$-dense within the ball $\ballM{p_i}{c\tsamconst_i}$, where
\begin{equation*}
  \tsamconst_i = \frac{\samconst(p_i)}{1 - c\nu}.
\end{equation*}
For any $p_j \in \ballM{p_i}{\tsamconst_i}$, we obtain
$\abs{\tsamconst_i - \tsamconst_j} \leq \lipsamconst \tsamconst_i$, where
\begin{equation}
  \label{eq:geom.lipsam}
  \lipsamconst = \frac{c \nu}{1 - c\nu},
\end{equation}
and if $\nu \leq \frac{1}{2c}$, then $\lipsamconst \leq 1$. 

\newcommand{\hsparseconst}{\hat{\mu}_0}
Similarly, if $\pts$ is $\hsparseconst \samconst$-separated, then it
will be $\tilde{\sparseconst} \tsamconst_i$-separated on $\ballM{p_i}{c
  \tsamconst_i}$, provided $\tilde{\sparseconst} \leq (1 -
2c\nu)\hat{\sparseconst}$. The constant $\hsparseconst$ itself
may be constrained to satisfy $\hsparseconst \leq (1 + \nu)^{-1}
\leq \frac{1}{2}$. 

In our framework here, the local constant sampling radii are applied to the local Euclidean metric, rather than the metric on the manifold, but the same idea applies.  Although \Eqnref{eq:geom.lipsam} indicates that $\lipsamconst$ is expected to become small as the sampling radius decreases, our analysis does not demand this. As explained in \Secref{sec:analysis.outline}, we only require that $\lipsamconst$ be mildly bounded. 

We summarise the assumptions on the input to the extended algorithm as Hypotheses~\ref{hyps:input}, where the parameters $\lipsamconst$ and $\metlipconst$ are left free to be constrained by subsequent hypotheses.

\begin{hyps}[Input assumptions]
  \label{hyps:input}
  We have a finite index set $\apts$ representing the sample points. For each $i \in \apts$ there is associated a subset of neighbours $\apts_i \subseteq \apts$. The geometry is imposed by
  \begin{enumerate}
  \item \textbf{Coordinate patches.} For each $i \in \apts$, there is a coordinate patch $U_i \subseteq \R^m$, and an injective application $\charti \colon \nbrsi \to U_i$ such that $\pts_i = \charti(\nbrsi)$ is a $(\sparseconst, \samconst_i)$-net for $\balli{p_i}{8\samconst_i} \subseteq U_i$. We introduce a parameter $\lipsamconst \geq 0$, and demand that, if $\disti{p_i}{p_j} \leq 6\samconst_i$, then
    \begin{equation*}
      \abs{\samconst_i - \samconst_j} \leq \lipsamconst
      \min\{\samconst_i,\samconst_j\}.
    \end{equation*}
  \item \textbf{Transition functions.}  Each $p_j \in \balli{p_i}{6\samconst_i} \subset U_i$ lies in the domain $U_{ij} \subseteq U_i$ of the transition function $\transij\colon U_{ij} \stackrel{\cong}{\longrightarrow}U_{ji}$. The domains must be suffiently large:
    \begin{equation*}
      \balli{p_i}{6\samconst_i} \cap \balli{p_j}{9\samconst_i} \subseteq U_{ij},
    \end{equation*}
    and the transition functions must satisfy the compatibility conditions $\transij = \trans{j}{i}^{-1}$ and $\transij \circ \charti = \chart{j}$ on $\nbrsi \cap \nbrs{j}$. Furthermore, the metric distortion of the transition functions is bounded by a parameter $\metlipconst$:
    \begin{equation*}
      \abs{\disti{x}{y} - \distj{\transij(x)}{\transij(y)}} \leq 
      \metlipconst \disti{x}{y}\quad \text{for all } x,y \in U_{ij}.
    \end{equation*}
  \end{enumerate}
\end{hyps}

\paragraph{The extended algorithm} The algorithm we present here is the same in spirit as the algorithm for the Euclidean setting~\cite{boissonnat2014flatpert} described in \Secref{sec:pert.strategy}, and we refer to it as the \defn{extended algorithm}.  It takes an input satisfying Hypotheses~\ref{hyps:input}. For each $i \in \apts$, a $\pertbnd\samconst_i$-perturbation is repeatedly applied to the point $p_i \in \pts_i' \subset U_i$  until a good perturbation $p_i'$ is found. The \Defref{def:extended.good.pert} of a good perturbation involves something closely resembling the hoop property~\ref{hyp:clean.hoop.bnd} with a parameter $\thoopbnd > \hoopbnd$. When a good perturbation $p'_i$ is selected, then the affected point sets $\ppts_j$ are updated, as well as $\ppts_i$ itself. By demonstrating stability of the hoop property with respect to small changes in the Euclidean metric, we are able to show that when the extended algorithm terminates, there will be no forbidden configurations in the region of interest of any local Euclidean coodinate patch. Then, assuming appropriate constraints on $\metlipconst$ and $\lipsamconst$, a manifold simplicial complex whose vertex set is $\apts$ is constructed by defining the star of $i$ to correspond to $\str{p'_i;\delof{\ppts_i}}$. The stability of these stars ensures that this complex is indeed a manifold (\Thmref{thm:man.mesh}), and we call it $\delof{\ppts}$, as justified by \Thmref{thm:output.del.cplx}. We refer to $\delof{\ppts}$ as the output of the extended algorithm, thus we assume that the extended algorithm includes a final step of computing all the stars after the perturbation algorithm has completed.

    % -*- LaTeX -*-
% analysis_outline.tex
% 20130614 extracted from manmesh.tex (created 20130420)
%
% proof of validity of the algorithm (with some details removed)

\subsection{Outline of the analysis}
\label{sec:analysis.outline}

%\section{Validity of the extended algorithm}

We have defined the point sets $\pts_i = \charti(\nbrsi)$ in the
coordinate patch for $p_i$. We will let $\ppts_i$ denote the
corresponding perturbed point set at any stage in the algorithm:
$\ppts_i$ changes during the course of the algorithm, and we do not
rename it according to the iteration as was done in the original
description of the algorithm for flat manifolds
\cite{boissonnat2014flatpert}. The perturbation of a point $p_i
\mapsto p'_i$ is performed in the coordinate patch $U_i$, and then all
the coordinate charts must be updated so that if $i \in \nbrs{j}$,
then $\pchart{j}(i) = \transij(p'_i)$. However, we will refer to the
point as $p'_i$ regardless of which coordinate frame we are
considering.  The discrete maps $\pcharti$ will change as the
algorithm progresses, but $\charti$ will always refer to the initial
map.

% In order to exploit \Lemref{lem:perturb.Delone}, we need to constrain the perturbation so that \Eqnref{eq:pertbnd} 
In order to ensure that we maintain a \pueset\ in each coordinate chart (\cite[Lemma 2.2]{boissonnat2014flatpert}), we need to constrain the point perturbation so that in any local Euclidean coordinate patch $U_j$, the cumulative perturbation is a $\tpertbnd\samconst_j$-perturbation with $\tpertbnd \leq \frac{\sparseconst}{4}$. If $p_i \mapsto \pert(p_i)$ such that $\disti{\pert(p_i)}{p_i} \leq \pertconst = \pertbnd \samconst_i$, then $\distj{\pert(p_i)}{p_i} \leq (1 + \metlipconst)\pertconst \leq (1 + \metlipconst)(1 + \lipsamconst)\pertbnd \samconst_j$, and we have
\begin{equation*}
  \tpertbnd = (1 + \metlipconst)(1 + \lipsamconst)\pertbnd.  
\end{equation*}
 Thus we demand that
\begin{equation}
  \label{eq:new.pertbnd}
%  \pertbnd \leq \left( (1 + \lipsamconst)(1 + \metlipconst) \right)^{-1}
  \tpertbnd \leq
  \frac{\sparseconst}{4}.
\end{equation}
In order to facilitate the analysis, we want an explicit constant to bound the ratio between $\pertbnd$ and $\tpertbnd$. We ensure that
\begin{equation}
  \label{eq:chart.change.bnd}
  (1 + \lipsamconst)(1 + \metlipconst) \leq 2
\end{equation}
by imposing the mild constraint that
\begin{equation}
  \label{eq:lipsambnd}
  \lipsamconst \leq \frac{1 - \metlipconst}{1 + \metlipconst}.
\end{equation}

We will keep the definition of forbidden configuration as in the flat
case. In other words a forbidden configuration is that which satisfies
the four properties described in \Lemref{lem:prop.forbid.cfg}, where
$\samconst$ refers to the local sampling radius $\samconst_i$. 

We do not attempt to remove the forbidden configurations from all of
$\ppts_i$. Rather, we define $\lptsi = \ppts_i \cap \balli{p_i}{6
  \samconst_i}$ as our region of interest. The reasoning behind this choice
appears in \Secref{sec:domain.size}, where we also show
%(\Lemref{lem:protected.stars.appdx}) 
that 
\cite[Lemma 3.6]{boissonnat2014flatpert} implies:
\begin{dup}[\Lemref{lem:protected.stars} (Protected stars)]
%  \label{lem:protected.stars}
  If there are no forbidden configurations in $\lptsi$, then all the
  $m$-simplices in $\str{p'_i; \delof{\lptsi} }$ are $\flakebnd$-good
  and $\delta$-protected, with $\delta = \delta_0 \psparseconst
  \psamconst_i$.
\end{dup}

This allows us to exploit the Delaunay metric stability result
\cite[Theorem 4.17]{boissonnat2013stab1}, which we show
(\Secref{sec:domain.size})
%(\Lemref{lem:stable.stars.appdx}) 
may be stated in our current context as:
\begin{dup}[\Lemref{lem:stable.stars} (Stable stars)]
%  \label{lem:stable.stars}
  If
  \begin{equation*}
    \metlipconst \leq \frac{\flakebnd^{2m+1}\sparseconst^2}{2^{12}},
  \end{equation*}
  and there are no forbidden configurations in $\lptsi$, then for all
  $p'_j \in \str{p'_i; \delof{\ppts_i}}$, we have
  \begin{equation*}
    \str{p'_i; \delof{\ppts_i}} \cong \str{p'_i; \delof{\ppts_j}}.
  \end{equation*}
\end{dup}

The main technical result we develop in the current analysis is the
bound on the distortion of the hoop property
%(\Lemref{lem:hoop.distort.appdx})
(\Secref{sec:hoop.distort}) due to the transition functions:
\begin{dup}[\Lemref{lem:hoop.distort} (Hoop distortion)]
%  \label{lem:hoop.distort}
  If
  \begin{equation*}
    \metlipconst \leq \left( \frac{\flakebnd^{2m + 1}}{4} \right)^2,
  \end{equation*}
  then for any forbidden configuration $\splxt =
  \splxjoin{p'_j}{\splxs} \subset \lptsi$, there is a simplex $\tsplxs
  = \transij(\splxs) \subset \ppts_j$ such that
  $\distj{p'_j}{\diasphere{\tsplxs}} \leq 2\thoopbnd\samconst_j$,
  where
 \begin{equation*}
    \thoopbnd =  \frac{2^{16} m^{\frac{3}{2}} \flakebnd}{\sparseconst^3}.
  \end{equation*}
\end{dup}
The proof of \Lemref{lem:hoop.distort} relies heavily on the thickness
bound (Property~\ref{hyp:good.facets}) for the facets of a forbidden
configuration. In \Secref{sec:splx.distort} we show bounds on the
changes of the intrinsic properties, such as thickness and
circumradius, of a Euclidean simplex subjected to the influence of a
transition function. This leads, as shown in
\Secref{sec:distortion.maps}, to bounds on circumcentre displacement
under small changes of a Euclidean metric. These bounds could not be
recovered directly from earlier work \cite{boissonnat2013stab1},
because they involve simplices that are not full dimensional. With
these results in place, the proof of \Lemref{lem:hoop.distort} is
assembled in \Secref{sec:hoop.distort}.

By considering the diameter of a forbidden configuration subjected to
metric distortion, we can determine the size of the neighbourhood of
$p_i$ that must be considered when checking whether a perturbation
$p_i \mapsto p'_i$ creates conflicts.

Suppose $\splxt \subset \lptsj \subset U_j$ is a forbidden configuration
with $p'_i \in \splxt$. By \Lemref{lem:prop.forbid.cfg},
Property~\ref{hyp:diam.bnd}, we have $\longedge{\splxt} <
\frac{5}{2}(1 + \frac{1}{2}\delta_0 \sparseconst) \samconst_j$. It
follows then that if $\tsplxt = \trans{j}{i}(\splxt) \subset \ppts_i$, then
\begin{equation*}
  \begin{split}
    \longedge{\tsplxt} &< (1 + \metlipconst)(1 + \lipsamconst)
    \frac{5}{2} \left( 1 + \frac{1}{2}\delta_0 \sparseconst \right)
    \samconst_i\\ 
    &\leq 5 \left(1 + \frac{1}{2}\delta_0 \sparseconst \right) \samconst_i,
 \end{split}
\end{equation*}
and we find, as in \cite[Lemma 4.4]{boissonnat2014flatpert}, that
if $\delta_0 \leq \frac{2}{5}$, then
\begin{equation*}
%  \label{eq:forbid.cfg.labels}
  (\pchart{j})^{-1}(\splxt) \subset \charti^{-1}(\balli{p_i}{r} \cap \pts_i ),
  \quad \text{where} \quad r = \left( 5 + \frac{3\sparseconst}{2}
  \right) \samconst_i. 
\end{equation*}
Indeed, this is ensured by the fact that $\pts_i$ is a $(\sparseconst,
\samconst_i)$-net for $\balli{p_i}{8\samconst_i}$, and $8\samconst_i -
r >  \samconst_i$. 

Let $\mathcal{S}_i$ denote all the $m$-simplices in $\apts_i$ whose
vertices are contained in
\begin{equation*}
  \charti^{-1}( \balli{p_i}{r} \cap \pts_i)
  \setminus \{i\} \quad \text{where} \quad r = \left(5 +
    \frac{3\sparseconst}{2} \right) \samconst_i.   
\end{equation*}
Then the simple packing argument demonstrated in \cite[Lemma
5.1]{boissonnat2014flatpert} yields
\begin{equation}
  \label{eq:bnd.nbr.cplx}
  \card{\mathcal{S}_i} < \left( \frac{14}{\sparseconst} \right)^{m^2 + m}.
\end{equation}

We strengthen the definition of a good perturbation:
\begin{de}[Good perturbation]
  \label{def:extended.good.pert}
  For the extended algorithm, we say that $p_i \mapsto x$ is a
  \defn{good perturbation} of $p_i \in U_i$ if there are no
  $m$-simplices $\splxs \in \pcharti(\mathcal{S}_i)$ such that
  $\disti{x}{\diasphere{\splxs}} \leq 2\thoopbnd\samconst_i$, where
  $\thoopbnd$ is defined in \Lemref{lem:hoop.distort}.
\end{de}
It is sufficient to only consider the $m$-simplices, because if
$\splxs$ is a non-degenerate $j$-simplex, with $j<m$, then it is the
face of some non-degenerate $m$-simplex $\splxt$, and
$\circsphere{\splxs} \subset \diasphere{\splxt}$.  With this
definition of a good perturbation, the extended algorithm yields the
analogue of \cite[Lemma 4.3]{boissonnat2014flatpert}:
%\Lemref{lem:no.forbidden.output}: 
\begin{lem}
  After the extended algorithm terminates there will be no forbidden
  configurations in $\lptsi$, for every $i \in \apts$.
\end{lem}
\begin{proof}
  We argue by induction that after the $i^{\text{th}}$ iteration, for
  any $j \leq i$, and any $k \in \apts$, there are no forbidden
  configurations in $\lpts_k$ that have $p'_j$ as a vertex.  For
  $i=1$, the assertion follows directly from
  \Defref{def:extended.good.pert}, and
  \Lemref{lem:hoop.distort}. Assume the assertion is true for
  $i-1$. Suppose $\splxt$ is a forbidden configuration in $\lpts_k$,
  after the $i^{\text{th}}$ iteration.  Then since $p'_i$ is a good
  perturbation, according to \Defref{def:extended.good.pert}, $\splxt$
  cannot contain $p'_i$. Also, $\splxt$ cannot contain any $p'_j$ with
  $j<i$, for that would contradict the induction hypothesis. Thus the
  hypothesis holds for all $i\geq 1$.
\end{proof}

We need to quantify the conditions under which the algorithm is
guaranteed to terminate. We use the same volumetric analysis that is
demonstrated in the proof of \cite[Lemma
5.4]{boissonnat2014flatpert},
%\Lemref{lem:good.perturbation.existence},
with the only modifications being a change in two of the constants
involved in the calculation. In particular, the number of simplices
involved is now given by \Eqnref{eq:bnd.nbr.cplx}, and we use the
bound on $\thoopbnd$ given by \Lemref{lem:hoop.distort}, which is
$2^3 m^{\frac{3}{2}}$ times the bound
on $\hoopbnd$ used in the original calculation. 
% Calc done in CF13p37
This calculation, coupled with the criterion for \Lemref{lem:hoop.distort}, yields a constraint on $\metlipconst$ with respect to the perturbation parameter $\pertbnd$. Together with Equations \eqref{eq:new.pertbnd} and \eqref{eq:lipsambnd}, this gives us all the constraints on the parameters that will ensure the existence of good perturbations, and therefore the termination of the algorithm:
\begin{hyps}[Parameter constraints]
  \label{hyps:param.bnds}
  Define
  \begin{equation*}
    \tpertbnd = (1 + \lipsamconst) (1 + \metlipconst) \pertbnd,
  \end{equation*}
  We require
  \begin{equation*}
    \lipsamconst \leq \frac{1 - \metlipconst}{1 + \metlipconst},
    \quad \text{and} \quad
    \tpertbnd \leq \frac{\sparseconst}{4},
    \quad \text{and} \quad
    \metlipconst \leq \frac{1}{2^4} \left(\frac{\pertbnd}{C} \right)^{4m + 2},
  \end{equation*}
  where
  \begin{equation*}
    C = m^{\frac{3}{2}}\left( \frac{2}{\sparseconst} \right)^{4m^2 + 5m + 21}. 
  \end{equation*}
\end{hyps}

Thus, using \Lemref{lem:protected.stars}, the main result
\cite[Theorem 4.1]{boissonnat2014flatpert} of the original
perturbation
%\Thmref{thm-main-theorem-of-the-paper}
algorithm can be adapted to the context of the extended algorithm as:
\begin{lem}[Algorithm guarantee]
  \label{lem:loc.eucl.protect}
  If Hypotheses~\ref{hyps:input} and \ref{hyps:param.bnds} are satisfied,
  % Let
  % \begin{equation*}
  %   \tpertbnd = (1 + \lipsamconst) (1 + \metlipconst) \pertbnd.
  % \end{equation*}
  % If 
  % \begin{equation*}
  %   \lipsamconst \leq \frac{1 - \metlipconst}{1 + \metlipconst},
  %   \quad \text{and} \quad
  %   \tpertbnd \leq \frac{\sparseconst}{4},
  %   \quad \text{and} \quad
  %   \metlipconst \leq \frac{1}{2^4} \left(\frac{\pertbnd}{C}
  %   \right)^{4m + 2}, 
  % \end{equation*}
  % where $C = m^{\frac{3}{2}}\left( \frac{2}{\sparseconst}
  % \right)^{4m^2 + 5m + 21}$,
  % % 
  then the extended algorithm terminates, and for every $i \in \apts$,
  the set $\lptsi$ is a $(\psparseconst, \psamconst_i)$-net such that
  there are no forbidden configurations with
  \begin{equation*}
    \flakebnd = \frac{\pertbnd}{C},
    \quad \text{and} \quad 
    \delta = \flakebnd^{m+1}\psparseconst \psamconst_i,
  \end{equation*}
  where $\psparseconst = \frac{\sparseconst-2\tpertbnd}{1+ \tpertbnd}$,
  and $\psamconst_i = (1+\tpertbnd)\samconst_i$.
\end{lem}

This allows us to apply \Lemref{lem:stable.stars}, and we can define
the abstract complex $\outmesh$ by the criterion that
$\pcharti(\str{i; \outmesh}) = \str{p'_i; \delof{\ppts_i}}$ for
all $i \in \apts$. This is a manifold piecewise linear\footnote{A manifold simplicial complex that admits an atlas of piecewise linear coordinate charts from the stars of the vertices is called \defn{piecewise linear}. There exists manifold simplical complexes that are not piecewise linear~\cite[Example 3.2.11]{thurston1997}, but they are not a concern for us here.} simplicial
complex.  The bound on $\metlipconst$ imposed by
\Lemref{lem:stable.stars} is met by the one imposed by
%\Lemref{lem:loc.eucl.protect},
Hypotheses~\ref{hyps:param.bnds}
and we arrive at our first main result:
% \begin{dup}[\Thmref{thm:man.mesh}~(Manifold mesh)]
%   Given an input satisfying Hypotheses~\ref{hyps:input} and \ref{hyps:param.bnds}, the extended algorithm produces a
%   manifold abstract simplicial complex $\outmesh$ defined by
%   \begin{equation*}
%     \str{i; \outmesh} \cong \str{p'_i; \delof{\ppts_i}}.
%   \end{equation*}
% \end{dup}
\begin{mainthm}[Manifold mesh]
  \label{thm:man.mesh}
  Given an input satisfying Hypotheses~\ref{hyps:input} and \ref{hyps:param.bnds}, the extended algorithm produces a
  manifold abstract simplicial complex $\outmesh$ defined by
  \begin{equation*}
    \str{i; \outmesh} \cong \str{p'_i; \delof{\ppts_i}}.
  \end{equation*}
  % Let
  % \begin{equation*}
  %   \tpertbnd = (1 + \lipsamconst) (1 + \metlipconst) \pertbnd.
  % \end{equation*}
  % If 
  % \begin{equation*}
  % \lipsamconst \leq \frac{1 - \metlipconst}{1 + \metlipconst},
  % \quad \text{and} \quad
  % \tpertbnd \leq \frac{\sparseconst}{4},
  %   \quad \text{and} \quad
  %   \metlipconst \leq \frac{1}{2^4} \left(\frac{\pertbnd}{C}
  %   \right)^{4m + 2}, 
  % \end{equation*}
  % where $C = m^{\frac{3}{2}}\left( \frac{2}{\sparseconst}
  % \right)^{4m^2 + 5m + 21}$, then the extended algorithm produces a
  % manifold abstract simplicial complex $\outmesh$ defined by
  % \begin{equation*}
  %   \str{i; \outmesh} \cong \str{p'_i; \delof{\ppts_i}}.
  % \end{equation*}
\end{mainthm}

The algorithm itself makes no explicit reference to the underlying manifold $\man$ or point set $\pts \subset \man$, but we need to consider these in order to justify the name $\outmesh$ for the output of the extended algorithm.

Given $\pts \subset \man$, we define the set $\ppts \subset \man$ to be the perturbed point set produced by the algorithm, i.e., $\pts \to \ppts$ is given by $p \mapsto p' = \fcharti^{-1}(p'_i)$, where $i \in \apts$ is the label associated with $p \in \pts$.  If the metric on $\man$ is such that the coordinate maps $\fcharti$ themselves have low metric distortion, then the constructed complex $\outmesh$ is in fact the Delaunay complex of $\ppts \subset \man$.  This follows from the fact that in the local Euclidean coordinate frames we have ensured that the points have stable Delaunay triangulations.  Thus, using $\flakebnd = \frac{\pertbnd}{C}$ given by \Lemref{lem:loc.eucl.protect}, the stability result \cite[Thm 4.17]{boissonnat2013stab1} 
%\Lemref{lem:del.metric.stability} 
leads, by the same reasoning that yields \Lemref{lem:stable.stars}, to the following:
\begin{mainthm}[Delaunay complex]
  \label{thm:output.del.cplx}
  Suppose that $\{(W_i,\fcharti)\}_{i \in \apts}$ is an atlas for the compact $m$-manifold $\man$, and the finite set $\pts \subset \man$ is such that
%  the conditions of \Thmref{thm:man.mesh}
%  \Secref{sec:input}
  Hypotheses~\ref{hyps:input} and \ref{hyps:param.bnds}
  are satisfied.  Suppose also that $\man$ is equipped with a metric $\gdistM$, such that
  \begin{equation*}
    \abs{\disti{\fcharti(x)}{\fcharti(y)} - \distM{x}{y} } \leq \eta
    \disti{\fcharti(x)}{\fcharti(y)},
  \end{equation*}
  whenever $x$ and $y$ belong to
  $\fcharti^{-1}(\ballEm{p_i}{6\samconst_i})$.
  If
  \begin{equation*}
    \eta \leq \frac{\sparseconst^2}{2^{12}} \left( \frac{\pertbnd}{C} \right)^{2m+1},
  \end{equation*}
  % and the conditions of \Thmref{thm:man.mesh} are met,
  then $\outmesh$ is the Delaunay complex of $\ppts \subset \man$ with respect to $\gdistM$.
\end{mainthm}
The required bound on $\eta$ is weaker than the bound on $\metlipconst$ demanded by Hypotheses~\ref{hyps:param.bnds}. In the standard scenario, the metric distortion of the transition functions is bounded by bounding the metric distortion of the coordinate charts, and in this case the bound on $\eta$ required by \Thmref{thm:output.del.cplx} is automatically met when Hypotheses~\ref{hyps:param.bnds} are satisfied.

    % -*- LaTeX -*-
% smooth.tex
% 20141020
%
% The Riemannian setting, (and general smooth case)

\subsection{The Riemannian setting} If $\man$ is a Riemannian manifold, then $\outmesh$ is a Delaunay triangulation of $\man$ and is equipped with a piecewise flat metric that is a good approximation of $\gdistM$. This follows from recent results~\cite{dyer2014riemsplx.arxiv} that guarantee a homeomorphism in this setting. 

We use the exponential map to define the coordinate charts. Proposition 17 and Lemma 11 of \cite{dyer2014riemsplx.arxiv} directly imply that if
\begin{equation*}
  \transij = \exp^{-1}_{\iota(j)} \circ \exp_{\iota(i)},
\end{equation*}
then on $\balli{p_i}{r}$ we have
\begin{equation*}
  \abs{ \distj{\transij(x)}{\transij(y)} - \disti{x}{y} }
  \leq 6 \Lambda r^2 \disti{x}{y},
\end{equation*}
where $\Lambda$ is a bound on the absolute value of the sectional curvatures of $\man$.  Here we will assume a constant sampling radius, i.e., $\lipsamconst = 0$ and $\samconst_j=\samconst_i$ for all $j \in \apts$. For our purposes, we need $r = 6\samconst_i$, and thus $\metlipconst = 6^3\Lambda \samconst_i^2$, and in order to satisfy Hypotheses~\ref{hyps:param.bnds} we require $6^3 \Lambda \samconst_i^2 \leq \frac{1}{2^4}(\flakebnd^{2m+1})^2$, which is satisfied if
\begin{equation}
  \label{eq:locsam.curv.bnd}
  \samconst_i \leq \frac{\flakebnd^{2m+1}}{2^6\sqrt{\Lambda}}.
\end{equation}

We exploit \cite[Theorem 3]{dyer2014riemsplx.arxiv}, which guarantees that the output complex is homeomorphic and with a metric quantifiably close to $\gdistM$. This demands that the star of each $p \in \pts \subset \man$ be contained in a geodesic ball $\ballM{p}{h}$ with
\begin{equation*}
  h = \min \left\{ \frac{\injradM}{4}, \frac{\flakebnd^m}{6\sqrt{\Lambda}} \right\}.
\end{equation*}
Since $\exp_p$ preserves the radius of a ball centred at $p$, we have that $h = 2\samconst_i$, and we see that the constraint $h \leq \frac{\flakebnd^m}{6\sqrt{\Lambda}}$ is automatically satisfied when $\samconst_i$ satisfies \eqref{eq:locsam.curv.bnd}.

We wish to express the required sampling conditions in terms of the intrinsic metric $\gdistM$. If $\samconst$ is the sampling radius with respect to $\gdistM$, we require an upper bound on $\samconst$ such that the needed bound on $\samconst_i$ is attained when accounting for metric distortion. The Rauch Theorem (\cite[Lemma 9]{dyer2014riemsplx.arxiv}) bounds the metric distortion of the exponential map, and it implies that within a ball of radius $r$
\begin{equation}
  \label{eq:exp.distort}
  \disti{\fcharti(x)}{\fcharti(y)} \leq \left( 1 + \frac{\Lambda r^2}{3} \right) \distM{x}{y}.
\end{equation}
In order to ensure that $\pts_i$ meets the density requirement of item 1 of Hypotheses~\ref{hyps:input}, we demand that $\balli{p_i}{9\samconst_i} \subseteq U_i$. Then, using \eqref{eq:locsam.curv.bnd} to bound $r=9\samconst_i$, we use \eqref{eq:exp.distort} to find the bound on $\samconst$ required to ensure \eqref{eq:locsam.curv.bnd}. In fact, the correction is so small that it is already accommodated by the adjustment made when we rounded the constant in \eqref{eq:locsam.curv.bnd} to a power of $2$. In other words, the right hand side of \eqref{eq:locsam.curv.bnd} is already sufficient as a bound on $\samconst$.

We also need to ensure that the conditions of Hypotheses~\ref{hyps:param.bnds} are met. In particular, if $\pts$ is a \ueset\ with respect to $\gdistM$, then the effective separation parameter with respect to $\gdisti$ will be slightly smaller, due to the metric distortion of the coordinate charts. In order to compute this correction, we again use the Rauch Theorem \cite[Lemma 9]{dyer2014riemsplx.arxiv}, and we find, for $p,q \in \pts$
\begin{equation*}
  \disti{\fcharti(p)}{\fcharti(q)}
  \geq \left(1 - \frac{\Lambda r^2}{2} \right)\sparseconst \samconst
  \geq \left(\frac{1 - \frac{\Lambda r^2}{2}}{1 + \frac{\Lambda r^2}{3}} \right)\sparseconst \samconst_i,
\end{equation*}
where $r = 9\samconst_i$, as above. Using \eqref{eq:locsam.curv.bnd}, and the constraint on $\flakebnd$ imposed by \Lemref{lem:loc.eucl.protect}, we find that the correction is indeed extremely small:
\begin{equation*}
  \disti{\fcharti(p)}{\fcharti(q)} \geq (1 - 2^{-200})\sparseconst \samconst_i.
\end{equation*}
We make crude adjustments to the constraint on $\pertbnd$ and the constant defining $\flakebnd$ to accommodate this. We can now formulate \cite[Theorem 3]{dyer2014riemsplx.arxiv} in our context:

\begin{mainthm}[Riemannian Delaunay triangulation]
  \label{thm:riem.del.tri}
  Suppose $\man$ is a Riemannian manifold, and $\pts \subset \man$ is 
a \ueset\ with respect to the metric $\gdistM$, with
\begin{equation*}
  \samconst \leq \min \left\{\frac{\injradM}{4},\frac{1}{2^6 \sqrt{\Lambda}} \left( \frac{\pertbnd}{\tilde{C}} \right)^{2m+1} \right\},
\end{equation*}
where $\Lambda$ is a bound on the absolute value of the sectional curvatures and $\injradM$ is the injectivity radius, and
\begin{equation*}
  \tilde{C} = m^{\frac{3}{2}}\left( \frac{2}{\sparseconst} \right)^{5m^2 + 5m + 21}.   
\end{equation*}
If the coordinate charts are defined by
\begin{equation*}
  \fcharti = \exp_{\iota(i)}^{-1}\colon \ballM{\iota(i)}{10\samconst} \to U_i,
\end{equation*}
and $\pertbnd \leq \frac{\sparseconst}{5}$, then the output $\outmesh$ of the extended algorithm is a Delaunay triangulation: there is a natural homeomorphism $H\colon\carrier{\outmesh} \to \man$ that satisfies
\begin{equation*}
  \abs{ \distM{H(x)}{H(y)} - \distG{\text{PL}}{x}{y} }
  \leq
  \left(2^8 \Lambda \left( \frac{\tilde{C}}{\pertbnd} \right)^{2m} \samconst^2 \right) \distG{\text{PL}}{x}{y},
\end{equation*}
where $\gdistG{\text{PL}}$ is the natural piecewise flat metric on $\outmesh$ defined by geodesic distances between vertices in $\man$. In addition, $\outmesh$ is self-Delaunay: it is a Delaunay triangulation of its vertices with respect to its intrinsic metric $\gdistG{\text{PL}}$.
\end{mainthm}
% The constants in the above theorem are calculated in my pencil notes 
% NB14p26 (20140420)
% 20141021 redone in NB14p32.

\paragraph{The general smooth case.} The homeomorphism result in the Riemannian setting can be exploited whenever the transition functions are smooth (or at least $C^3$), even if there is no explicit Riemannian metric associated with the input. The reason is that we can construct a Riemannian metric on the manifold by the standard trick using a partition of unity subordinate to the atlas~\cite[Thm. V.4.5]{boothby1986}. Then a short exercise shows that the metric distortion of the coordinate charts is bounded by $\metlipconst$.
% Work in U_i. Upper bound d_M(x,y) by finding an upper bound for the geodesic length of a straight line between x and y. Lower bound d_M(x,y) by a similar trick, but using a *geodesic* as the curve in the parameter domain, to find that $d_M(x,y) \geq (1 - \xi_0)\ell_i(\gamma)$, where $\ell_i(\gamma)$ is the length of gamma with respect to the metric d_i. This is of course larger than d_i(x,y).
It follows then that the constructed Riemannian metric satisfies \Thmref{thm:output.del.cplx} if Hypotheses~\ref{hyps:input} and \ref{hyps:param.bnds} are satisfied. Thus we can guarantee the existence of a Delaunay triangulation with respect to any smooth metric that can be locally approximated by a Euclidean metric to any desired accuracy (i.e., with arbitrarily small metric distortion): At a sufficiently high sampling density $\outmesh$ will satisfy \Thmref{thm:output.del.cplx} with respect to the given metric, as well as both of Theorems~\ref{thm:output.del.cplx} and \ref{thm:riem.del.tri} for the constructed Riemannian metric. 

The primary example of such a smooth, non-Riemannian metric is the metric of the ambient space $\amb$ restricted to a submanifold $\man \subset \amb$. The associated Delaunay complex is often called the restricted Delaunay complex. Sampling conditions that ensure that the Delaunay complexes associated with the restricted ambient metric and the induced Riemannian metric coincide are worked out in detail from the extrinsic point of view in \cite{boissonnat2013stab2.inria}.

%\input{review}
%\input{overview}
%    \input{analysis_outline}
%\input{quality_summary}

% -*- LaTeX -*-
% splx_distort.tex
% 20130420
%
% Results about simplices subjected to a distortion

\section{Details of the analysis}
\label{sec:details}

In this section we provide details to support the argument made in
\Secref{sec:analysis.outline}. 

\subsection{Simplex distortion}
\label{sec:splx.distort}

Our transition functions introduce a metric distortion when we move
from one coordinate chart to another. The geometric properties of a
simplex will be slightly different if we consider it with respect to
the Euclidean metric $\gdisti$ than they would be if we are using a
different Euclidean metric $\gdistj$. We wish to bound the magnitude
of the change of such properties as the thickness and the circumradius
of a simplex that is subjected to such a distortion. This is an
exercise in linear algebra.

We wish to compare two Euclidean simplices with corresponding vertices, but whose corresponding edge lengths differ by a relatively small amount. The embedding of the simplex in Euclidean space (i.e., the coordinates of the vertices) is not relevant to us. Previous results often only consider the case where the vertices of a given simplex are perturbed a small amount to obtain a new simplex. \Lemref{lem:splx.good.isometry} demonstrates the existence of an isometry that allows us to also consider the general situation in terms of vertex displacements.

We will exploit observations on the linear algebra of simplices
developed in previous work \cite{boissonnat2013stab1}. A $k$-simplex
$\splxs = \asimplex{p_0, \ldots, p_k}$ in $\rem$ can be represented by
an $m \times k$ matrix $P$, whose $i^{\text{th}}$ column is $p_i -
p_0$. We let $\sing{i}{A}$ denote the $i^{\text{th}}$ singular value
of a matrix $A$, and observe that $\norm{P} = \sing{1}{P} \leq
\sqrt{k}\longedge{\splxs}$. 

We are particularly interested in bounds on the smallest singular value of $P$, which is the inverse of the largest singular value of the pseudo-invese $\pseudoinv{P} = (\transp{P}P)^{-1}\transp{P}$. If the columns of $P$ are viewed as a basis for $\affhull{\splxs}$, then the rows of $\pseudoinv{P}$ may be viewed as the dual basis.  The magnitude of a dual vector is equal to the inverse of the corresponding altitudes in $\splxs$, and this leads directly to the desired bound on the smallest singular value of $P$, which is expressed in the following Lemma \cite[Lemma 2.4]{boissonnat2013stab1}:
\begin{lem}[Thickness and singular value]
  \label{lem:bound.skP}
  Let $\splxs = \simplex{p_0, \ldots, p_k}$ be a non-de\-generate
  $k$-simplex in $\rem$, with $k>0$, and let $P$ be the $m \times k$
  matrix whose $i^{\text{th}}$ column is $p_i - p_0$. Then
  the
  $i^{\text{th}}$ row of $\pseudoinv{P}$ is given by $\transp{w}_i$,
  where $w_i$ is orthogonal to $\affhull{\opface{p_i}{\splxs}}$, and
  \begin{equation*}
    \norm{w_i} = \splxalt{p_i}{\splxs}^{-1}.
  \end{equation*}
  We
  have the following bound on the smallest singular value of $P$:
  \begin{equation*}
    % \label{eq:bound.skP}
%    j \thickness{\splxs} \longedge{\splxs} \geq
    \sing{k}{P} \geq \sqrt{k} \thickness{\splxs}\longedge{\splxs}.
  \end{equation*}
\end{lem}

We will also have use for a lower bound on the thickness of $\splxs$,
given the smallest singular value for the representative matrix
$P$. We observe that $P$ was constructed by arbitrarily choosing one
vertex, $p_0$, to serve as the origin. If there is a vertex $p_i$,
different from $p_0$, such that $\splxalt{p_i}{\splxs}$ is minimal
amongst all the altitudes of $\splxs$, then according to
\Lemref{lem:bound.skP}, $\norm{w_i} = ( k \thickness{\splxs}
\longedge{\splxs} )^{-1}$, and it follows then that
$\sing{1}{\pseudoinv{P}} \geq ( k \thickness{\splxs} \longedge{\splxs}
)^{-1}$, and therefore
\begin{equation}
  \label{eq:upper.bnd.small.sing}
  \sing{k}{P} \leq k \thickness{\splxs} \longedge{\splxs},
\end{equation}
in this case. 

We are going to be interested here in purely intrinsic properties of
simplices in $\rem$; properties that are not dependent on the choice
of embedding in $\rem$. In this context it is convenient to make use
of the \defn{Gram matrix} $\transp{P}P$, because if $\transp{Q}Q =
\transp{P}P$, then there is an orthogonal transformation $O$ such that
$P = OQ$. This assertion becomes evident when considering the singular
value decompositions of $P$ and $Q$. Indeed, the entries of the Gram
matrix can be expressed in terms of squared edge lengths, as observed
in the proof of the following:
\begin{lem}
  \label{lem:splx.gram}
  Suppose that $\splxs = \asimplex{p_0, \ldots, p_k}$ and $\tsplxs =
  \asimplex{\tilde{p}_0, \ldots, \tilde{p}_k}$ are two $k$-simplices
  in $\rem$ such that
  \begin{equation*}
    \abs{\norm{p_i - p_j} - \norm{\tilde{p}_i - \tilde{p}_j}} \leq
    \metlipconst \longedge{\splxs},
  \end{equation*}
  for all $0 \leq i < j \leq k$.  Let $P$ be the matrix whose
  $i^{\text{th}}$ column is $p_i - p_0$, and define $\tilde{P}$
  similarly. Consider the Gram matrices, and let $E$ be the matrix
  that records their difference:
  \begin{equation*}
    \transp{\tilde{P}} \tilde{P} = \transp{P} P + E.
  \end{equation*}
  If $\metlipconst \leq \frac{2}{3}$, then the entries of $E$ are
  bounded by $\abs{E_{ij}} \leq 4 \metlipconst \longedge{\splxs}^2$,
  and in particular
  \begin{equation}
    \label{eq:gram.err.bnd}
    \norm{E} \leq 4 k \metlipconst \longedge{\splxs}^2.     
  \end{equation}
\end{lem}
\begin{proof}
  Let $v_i = p_i - p_0$, and $\tilde{v}_i = \tilde{p}_i -
  \tilde{p}_0$.  Expanding scalar products of the form
  $\transp{(v_j - v_i)}(v_j - v_i)$, we obtain a bound on the
  magnitude of the coefficients of $E$:
  \begin{equation*}
    \begin{split}
      \abs{ \transp{\tilde{v}}_i \tilde{v}_j - \transp{v}_i v_j }
      &\leq
      \frac{1}{2} \left( \abs{ \norm{\tilde{v}_i}^2 - \norm{v_i}^2 }
        + \abs{ \norm{\tilde{v}_j}^2 - \norm{v_j}^2 }
        + \abs{ \norm{\tilde{v}_j - \tilde{v}_i}^2 - \norm{v_j - v_i}^2 }
      \right) \\
      &\leq \frac{3}{2}(2 + \metlipconst)\metlipconst
      \longedge{\splxs}^2\\
      &\leq 4 \metlipconst \longedge{\splxs}^2.
    \end{split}
  \end{equation*}

  This leads us to a bound on $\sing{1}{E} = \norm{E}$. Indeed, the
  magnitude of the column vectors of $E$ is bounded by $\sqrt{k}$
  times a bound on the magnitude of their coefficients, and the
  magnitude of $\sing{1}{E}$ is bounded by $\sqrt{k}$ times a bound on
  the magnitude of the column vectors. We obtain
  \Eqnref{eq:gram.err.bnd}.
  % \begin{equation*}
  %   \sing{1}{E} \leq 4k \metlipconst \longedge{\splxs}^2.
  % \end{equation*}
\end{proof}

\Lemref{lem:splx.gram} enables us to bound the thickness of a
distorted simplex:
\begin{lem}[Thickness under distortion]
  \label{lem:intrinsic.thick.distortion}
  Suppose that $\splxs = \asimplex{p_0, \ldots, p_k}$ and $\tsplxs =
  \asimplex{\tilde{p}_0, \ldots, \tilde{p}_k}$ are two $k$-simplices in
  $\rem$ such that
  \begin{equation*}
    \abs{\norm{p_i - p_j} - \norm{\tilde{p}_i - \tilde{p}_j}} \leq
    \metlipconst \longedge{\splxs}
  \end{equation*}
  for all $0 \leq i < j \leq k$.  Let $P$ be the matrix whose
  $i^{\text{th}}$ column is $p_i - p_0$, and define $\tilde{P}$
  similarly.
 
  If
  \begin{equation*}
    \metlipconst \leq \left( \frac{\eta \thickness{\splxs}}{2} \right)^2
   \qquad \text{ with }  \eta^2 \leq 1,
  \end{equation*}
  then
  \begin{equation*}
    \sing{k}{\tilde{P}} \geq (1 - \eta^2) \sing{k}{P}, 
  \end{equation*}
  and
  \begin{equation*}
    \thickness{\tsplxs}\longedge{\tsplxs} \geq \frac{1}{\sqrt{k}}(1 - \eta^2)
    \thickness{\splxs}\longedge{\splxs},
  \end{equation*}
  and
  \begin{equation*}
    \thickness{\tsplxs} \geq \frac{4}{5\sqrt{k}}(1 - \eta^2)
    \thickness{\splxs}. 
  \end{equation*}
\end{lem}
\begin{proof}
  The equation $\transp{\tilde{P}}\tilde{P} = \transp{P} P + E$
  implies that
  \begin{equation*}
    \abs{\sing{k}{\tilde{P}}^2 - \sing{k}{P}^2}  \leq \sing{1}{E},
  \end{equation*}
  and so
  \begin{equation*}
    \abs{\sing{k}{\tilde{P}} - \sing{k}{P}} \leq
    \frac{\sing{1}{E}}{\sing{k}{\tilde{P}} + \sing{k}{P}}
    \leq
    \frac{\sing{1}{E}}{\sing{k}{P}}.
  \end{equation*}
  Thus
  \begin{equation*}
    \sing{k}{\tilde{P}} \geq \sing{k}{P} -
    \frac{\sing{1}{E}}{\sing{k}{P}} = \sing{k}{P} \left( 1 -
      \frac{\sing{1}{E}}{\sing{k}{P}^2} \right).
  \end{equation*}
  From \Lemref{lem:splx.gram} and the bound on $\metlipconst$ we have
  \begin{equation*}
    \sing{1}{E} \leq \eta^2 k \thickness{\splxs}^2 \longedge{\splxs}^2,
  \end{equation*}
  and so $\frac{\sing{1}{E}}{\sing{k}{P}^2} \leq \eta^2$ by
  \Lemref{lem:bound.skP}, and we obtain $\sing{k}{\tilde{P}} \geq (1 -
  \eta^2)\sing{k}{P}$. 

  For the thickness bound we assume, without loss of generality, that
  there is some vertex different from $\tilde{p}_0$ that realises the
  minimal altitude in $\tilde{\splxs}$ (our choice of ordering of the
  vertices is unimportant, other than to establish the correspondence
  between $\splxs$ and $\tsplxs$). Thus
  \Eqnref{eq:upper.bnd.small.sing} and  \Lemref{lem:bound.skP}, give
  the inequalities
  \begin{equation*}
    k \thickness{\tsplxs}\longedge{\tsplxs} \geq \sing{k}{\tilde{P}},
    \qquad \text{ and } \qquad
    \sing{k}{P} \geq \sqrt{k} \thickness{\splxs}\longedge{\splxs},
  \end{equation*}
  and we obtain
  \begin{equation*}
    k \thickness{\tsplxs} \longedge{\tsplxs} \geq 
    (1 - \eta^2) \sqrt{k}\thickness{\splxs} \longedge{\splxs}.
  \end{equation*}
  The final result follows since
  $\frac{\longedge{\splxs}}{\longedge{\tsplxs}} \geq \frac{1}{1 +
    \metlipconst} \geq \frac{4}{5}$.
\end{proof}

In order to obtain a bound on the circumradius of $\tsplxs$ with
respect to that of $\splxs$, it is convenient to find an isometry that
maps the vertices of $\splxs$ close to the vertices of
$\tsplxs$. Choosing $\tilde{p}_0$ and $p_0$ to coincide at the origin,
the displacement error for the remaining vertices is minimised by
taking the orthogonal polar factor of the linear transformation $A=
\tilde{P}P^{-1}$ that maps $\splxs$ to $\tsplxs$. In other words, if
the singular value decomposition of $A$ is $A = U_A \Sigma_A
\transp{V}_A$, then $A = \Phi S$, where $S = V_A \Sigma_A
\transp{V}_A$, and $\Phi = U_A \transp{V}_A$ is the desired linear
isometry. We have the following result, which is a special case of a
theorem demonstrated by Jim\'enez and Petrova~\cite{jimenez2013}:
\begin{lem}[Close alignment of bases]
  \label{lem:good.isometry}
  Suppose that $P$ and $\tilde{P}$ are non-degenerate $k \times k$
  matrices such that
  \begin{equation}
    \label{eq:gram.perturb}
    \transp{\tilde{P}}\tilde{P} = \transp{P} P + E.
  \end{equation}
  Then there exists a linear isometry $\Phi: \reel^k \to \reel^k$ such
  that 
  \begin{equation*}
    \norm{\tilde{P} - \Phi P} \leq \frac{\sing{1}{P}
      \sing{1}{E}}{\sing{k}{P}^2}. 
  \end{equation*}
\end{lem}
\begin{proof}
  Multiplying by $\invtransp{P} := \inv{ ( \transp{P})}$ on the left,
  and by $\inv{P}$ on the right, we rewrite \Eqnref{eq:gram.perturb} as
  \begin{equation}
    \label{eq:A.F}
   \transp{A}A = I + F,
  \end{equation}
  where $A = \tilde{P} \inv{P}$, and $F = \invtransp{P}E \inv{P}$.
  Using the singular value decomposition $A = U_A \Sigma_A
  \transp{V}_A$, we let $\Phi = U_A \transp{V}_A$, and we find
  \begin{equation}
    \label{eq:expand.mat.diff}
    \tilde{P} - \Phi P = (A - \Phi) P = U_A( \Sigma_A - I )
    \transp{V}_A P.
  \end{equation}
  From \Eqnref{eq:A.F} we deduce that $\sing{1}{A}^2 \leq 1 +
  \sing{1}{F}$, and also that $\sing{k}{A}^2 \geq 1 - \sing{1}{F}$. It
  follows that
  \begin{equation*}
    \max_i \abs{\sing{i}{A} - 1} \leq \frac{\sing{1}{F}}{1 +
      \sing{i}{A}} \leq \sing{1}{F},
  \end{equation*}
  and thus
 \begin{equation*}
    \norm{\Sigma_A - I} \leq \sing{1}{F} \leq \sing{1}{\inv{P}}^2
    \sing{1}{E} = \sing{k}{P}^{-2} \sing{1}{E}.
  \end{equation*}
  The result now follows from \Eqnref{eq:expand.mat.diff}.
\end{proof}
Recalling that an upper bound on the norm of a matrix also serves as
an upper bound on the norm of its column vectors, we obtain the
following immediate consequence of \Lemref{lem:good.isometry}, using
\Lemref{lem:splx.gram} and \Lemref{lem:bound.skP}:
\begin{lem}[Close alignment of simplices]
  \label{lem:splx.good.isometry}
  Suppose that $\splxs = \asimplex{p_0, \ldots, p_k}$ and $\tsplxs =
  \asimplex{\tilde{p}_0, \ldots, \tilde{p}_k}$ are two $k$-simplices
  in $\rem$ such that
  \begin{equation*}
    \abs{\norm{p_i - p_j} - \norm{\tilde{p}_i - \tilde{p}_j}} \leq
    \metlipconst \longedge{\splxs},
  \end{equation*}
  for all $0 \leq i < j \leq k$.  Let $P$ be the matrix whose
  $i^{\text{th}}$ column is $p_i - p_0$, and define $\tilde{P}$
  similarly. If $\metlipconst \leq \frac{2}{3}$, then there exists an
  isometry $\Phi: \rem \to \rem$ such that
  \begin{equation*}
    \norm{\tilde{P} - \Phi P} \leq \frac{4 \sqrt{k} \metlipconst
      \longedge{\splxs} }{\thickness{\splxs}^2},
  \end{equation*}
  and if $\hat{\splxs} = \Phi \splxs = \asimplex{\hat{p}_0,
    \ldots, \hat{p}_k}$, then $\hat{p}_0 = \tilde{p}_0$, and
  \begin{equation*}
    \norm{\hat{p}_i - \tilde{p}_i} \leq \frac{4 \sqrt{k} \metlipconst
      \longedge{\splxs} }{\thickness{\splxs}^2} \qquad \text{for all }
    1 \leq i \leq k.
  \end{equation*}
\end{lem}
Using \Lemref{lem:splx.good.isometry} together with
\cite[Lemma 4.3]{boissonnat2013stab1} we obtain a bound on the difference
in the circumradii of two simplices whose edge lengths are almost the same:
\begin{lem}[Circumradii under distortion]
  \label{lem:circ.rad.distort}
  Suppose that $\splxs = \asimplex{p_0, \ldots, p_k}$ and $\tsplxs =
  \asimplex{\tilde{p}_0, \ldots, \tilde{p}_k}$ are two $k$-simplices
  in $\rem$ such that
  \begin{equation*}
    \abs{\norm{p_i - p_j} - \norm{\tilde{p}_i - \tilde{p}_j}} \leq
    \metlipconst \longedge{\splxs},
  \end{equation*}
  for all $0 \leq i < j \leq k$. If
  \begin{equation*}
    \metlipconst \leq \left( \frac{\thickness{\splxs}}{4}
    \right)^2, %\qquad \text{with } \eta^2 \leq \frac{1}{2},
  \end{equation*}
  then
  \begin{equation*}
    \abs{\circrad{\tsplxs} - \circrad{\splxs} } \leq \frac{16
      k^{\frac{3}{2}} \circrad{\splxs} \metlipconst }{\thickness{\splxs}^3 }.
  \end{equation*}
\end{lem}
\begin{proof}
  We define $\hat{\splxs} = \Phi \splxs$, where $\Phi: \splxs \to
  \affhull{\tsplxs}$ is the isometry described in
  \Lemref{lem:splx.good.isometry}.  Since $\hat{p}_0 = \tilde{p}_0$,
  and $\circrad{\hat{\splxs}} = \circrad{\splxs}$, we have
  $\abs{\circrad{\tsplxs} - \circrad{\splxs}} \leq
  \norm{\circcentre{\hat{\splxs}} - \circcentre{\tsplxs}}$.  By
  \Lemref{lem:good.isometry}, the distances between
  $\circcentre{\splxs}$ and the vertices of $\tsplxs$ are all bounded
  by
  \begin{equation*}
    \circrad{\splxs} + \frac{4 \sqrt{k} \metlipconst
      \longedge{\splxs} }{\thickness{\splxs}^2}
    \leq
    (1 + \frac{\sqrt{k}}{2}) \circrad{\splxs} \leq \frac{3
      \sqrt{k}}{2} \circrad{\splxs},
  \end{equation*}
  and these distances differ by no more than
  \begin{equation*}
    \frac{ 8 \sqrt{k} \metlipconst \longedge{\splxs}
    }{\thickness{\splxs}^2}. % \leq 2\sqrt{k} \eta^2 \longedge{\splxs}.
  \end{equation*}
  It follows then from \cite[Lemma 4.3]{boissonnat2013stab1} that
  \begin{equation*}
    \begin{split}
      \norm{\circcentre{\hat{\splxs}} - \circcentre{\tsplxs}}
      &\leq
      \frac{ \frac{3 \sqrt{k}}{2}\circrad{\splxs}} 
      { \thickness{\tsplxs} \longedge{\tsplxs} }
      \left(     \frac{ 8 \sqrt{k} \metlipconst \longedge{\splxs}
        }{\thickness{\splxs}^2} \right)\\
      &\leq
      \frac{ 12 k\circrad{\splxs} \metlipconst } 
      { \frac{3}{4\sqrt{k}}\thickness{\splxs}^3 } \qquad
      \text{by
        \Lemref{lem:intrinsic.thick.distortion}, with } \eta = \frac{1}{2}\\
      &\leq
      \frac{ 16 k^{\frac{3}{2}} 
        \circrad{\splxs} \metlipconst} 
      {  \thickness{\splxs}^3 }.
    \end{split}
  \end{equation*}
\end{proof}

\subsection{Circumcentres and distortion maps}
\label{sec:distortion.maps}

It is convenient to introduce the affine space $\normhull{\splxs}$,
which is the space of centres of circumscribing balls for a simplex
$\splxs \in \rem$. If $\splxs$ is a non-degenerate $k$-simplex, then
$\normhull{\splxs}$ is an affine space of dimension $m-k$
perpendicular to $\affhull{\splxs}$ and containing
$\circcentre{\splxs}$.

The transition functions introduce a small metric distortion, which
motivated our interest in the properties of perturbed simplices. In
order to extend the perturbation algorithm
\cite{boissonnat2014flatpert} to the setting of curved
manifolds, we are interested in quantifying how the test for the hoop
property behaves under a perturbation of the interpoint
distances. Specifically, if a point $p$ is at a distance $\hoopbnd R$
from the diametric sphere of a simplex $\splxs$ in one coordinate
frame, what can we say about the distance of $p$ from
$\diasphere{\splxs}$ when measured by the metric of another coordinate
frame? To this end, we are interested in the behaviour of the
circumcentre under the influence of a mapping that is not distance
preserving. As a first step in this direction, we observe another
consequence of \cite[Lemma 4.3]{boissonnat2013stab1}:
\begin{lem}[Circumscribing balls under distortion]
  \label{lem:circ.ball.distort}
  Suppose $\phi: \rem \supset U \to V \subset \rem$ is a homeomorphism
  such that, for some positive $\metlipconst$,
  \begin{equation*}
    \abs{ \distEm{x}{y} - \distEm{\phi(x)}{\phi(y)} } \leq 
    \metlipconst \distEm{x}{y} \qquad \text{ for all } \quad x,y \in U.
  \end{equation*}
  Suppose also that $\splxs \subset U$ is a $k$-simplex, and that
  $\ballEm{c}{r}$ is a circumscribing ball for $\splxs$ with $c \in
  U$. Let $\tsplxs = \phi(\splxs)$. If
  \begin{equation*}
    \metlipconst \leq \left( \frac{\thickness{\splxs}}{4} \right)^2,
  \end{equation*}
  then there is a circumscribing ball $\ballEm{\tilde{c}}{\tilde{r}}$
  for $\tsplxs$ such that
  \begin{equation*}
    \distEm{\phi(c)}{\tilde{c}} \leq \frac{3 \sqrt{k}r^2
      \metlipconst}{\thickness{\splxs} \longedge{\splxs}}, 
  \end{equation*}
  and
  \begin{equation*}
    \abs{\tilde{r} - r } \leq \frac{5 \sqrt{k}r^2
      \metlipconst}{\thickness{\splxs} \longedge{\splxs}}.
  \end{equation*}
\end{lem}
\begin{proof}
  By the perturbation bounds on $\phi$, the distances between
  $\phi(c)$ and the vertices of $\tsplxs$ differ by no more than
  $2\metlipconst r$, and these distances are all bounded by $(1 +
  \metlipconst)r$. In this context \cite[Lemma 4.3]{boissonnat2013stab1}
  says that there exists a $\tilde{c} \in \normhull{\splxs}$ such that
  \begin{equation*}
    \distEm{\phi(c)}{\tilde{c}} \leq \frac{(1 + \metlipconst)r
      2\metlipconst r}{\thickness{\tsplxs}\longedge{\tsplxs}}. 
  \end{equation*}
  We apply \Lemref{lem:intrinsic.thick.distortion}, using $\eta =
  \frac{1}{2}$, to obtain
 $\thickness{\tsplxs}\longedge{\tsplxs} \geq
 \frac{3}{4\sqrt{k}}\thickness{\splxs} \longedge{\splxs}$. We find
  \begin{equation*}
    \distEm{\phi(c)}{\tilde{c}} \leq \frac{8 \sqrt{k} (1 + \metlipconst)r^2
      \metlipconst }{ 3\thickness{\splxs}\longedge{\splxs}}. 
  \end{equation*}
  The announced bound on $\distEm{\phi(c)}{\tilde{c}}$ is obtained by
  observing that $\metlipconst \leq \frac{1}{16}$. 

  Choosing a vertex $\tilde{p} = \phi(p) \in \tsplxs$, the bound on
  the difference in the radii follows:
  \begin{equation*}
    \begin{split}
      \tilde{r} = \distEm{\tilde{p}}{\tilde{c}}
      &\geq \distEm{\tilde{p}}{\phi(c)} - \distEm{\phi(c)}{\tilde{c}}\\
      &\geq r - \metlipconst r - \frac{3 \sqrt{k}r^2
        \metlipconst}{\thickness{\splxs} \longedge{\splxs}}\\ 
      &\geq r  - \frac{5 \sqrt{k}r^2
        \metlipconst}{\thickness{\splxs} \longedge{\splxs}},
    \end{split}
  \end{equation*}
  and similarly for the upper bound.
\end{proof}

We will find it convenient to have a bound on the circumradius of a
simplex, relative to its thickness and longest edge length:
\begin{lem}
  \label{lem:crude.rad.bnd}
  If $\splxs$ is a non-degenerate simplex in $\rem$, then
  \begin{equation*}
    \circrad{\splxs} \leq \frac{\longedge{\splxs}}{2 \thickness{\splxs}}.
  \end{equation*}
\end{lem}
\begin{proof}
  Let $\splxs = \asimplex{p_0, \ldots, p_k}$, We work in $\reel^k =
  \affhull{\splxs} \subset \rem$, and let $P$ be the $k \times k$
  matrix whose $i^{\text{th}}$ column is $p_i - p_0$. Then, by
  equating $\norm{\circcentre{\splxs} - p_0}^2$ with
  $\norm{\circcentre{\splxs} - p_i}^2$ and expanding, we find a system
  of equations that may be written in matrix form as
  \begin{equation*}
    \transp{P}\circcentre{\splxs} = b,
  \end{equation*}
  where the $i^{\text{th}}$ component of the vector $b$ is
  $\frac{1}{2}(\norm{p_i}^2 - \norm{p_0}^2)$. Choosing $p_0$ as the
  origin, we have $\norm{\circcentre{\splxs}} = \circrad{\splxs}$, and
  $\norm{b} \leq \frac{1}{2}\sqrt{k}\longedge{\splxs}^2$. Since
  $\sing{1}{\invtransp{P}} = \sing{k}{P}^{-1}$, the result follows
  from \Lemref{lem:bound.skP}, which says $\sing{k}{P}
  \geq \sqrt{k} \thickness{\splxs} \longedge{\splxs}$.
\end{proof}

Using the bound on
$\distEm{\phi(\circcentre{\splxs})}{\normhull{\tsplxs}}$ given by
\Lemref{lem:circ.ball.distort}, together with the circumradius bound
of \Lemref{lem:circ.rad.distort}, we obtain a bound on
$\distEm{\phi(\circcentre{\splxs})}{\circcentre{\tsplxs}}$ by means of
the Pythaogrean theorem:
\begin{lem}[Circumcentres under distortion]
  \label{lem:bound.cc}
  Suppose $\phi: \rem \supset U \to V \subset \rem$ is a homeomorphism
  such that
  \begin{equation*}
    \abs{ \distEm{x}{y} - \distEm{\phi(x)}{\phi(y)} } \leq 
    \metlipconst \distEm{x}{y} \qquad \text{ for all } \quad x,y \in
    U. 
  \end{equation*}
  Suppose also that $\splxs \subset U$ is a $k$-simplex, and let
  $\tsplxs = \phi(\splxs)$. If
  \begin{equation*}
    \metlipconst \leq \left( \frac{\thickness{\splxs}}{4}
    \right)^2, 
 \end{equation*}
  then
  \begin{equation*}
    \distEm{\phi(\circcentre{\splxs})}{\circcentre{\tsplxs}}
    \leq \left[ \left( \frac{42 k^2}{\thickness{\splxs}^3}
      \right)  \metlipconst
    \right]^{\frac{1}{2}} \circrad{\splxs}.
  \end{equation*}
\end{lem}
\begin{figure}[ht]
  \begin{center}
    \includegraphics[width=.7\columnwidth]{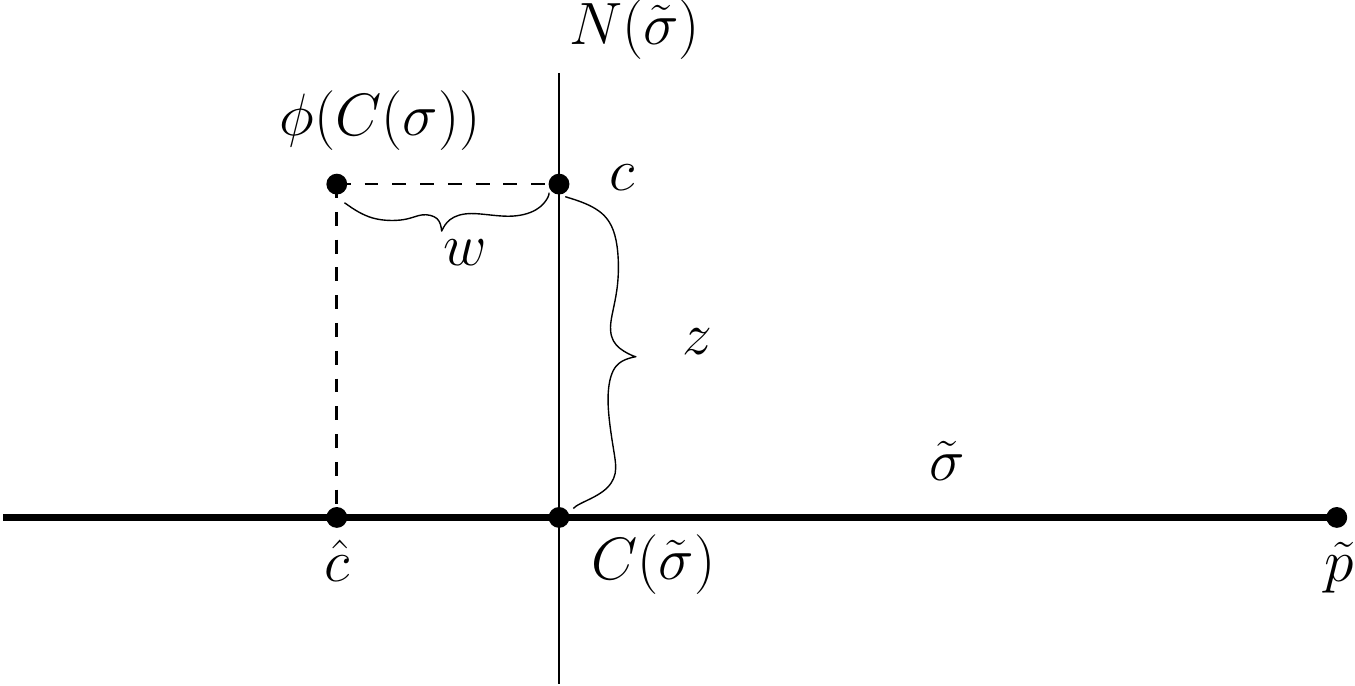} 
  \end{center}
  \caption{Diagram for the proof of \Lemref{lem:bound.cc}.}
  \label{fig:bound.cc}
\end{figure}
\newcommand{\tR}{\tilde{R}}
\begin{proof}
  Let $c$ be the closest point in $\normhull{\tsplxs}$ to
  $\phi(\circcentre{\splxs})$, and let $w$ be the distance from $c$ to
  $\phi(\circcentre{\splxs})$. Setting $z$ as the distance between $c$
  and $\circcentre{\tsplxs}$, we have that
  $\distEm{\phi(\circcentre{\splxs})}{\circcentre{\tsplxs}}^2 = z^2 +
  w^2$; see \Figref{fig:bound.cc}. Let $\hat{c}$ be the orthogonal
  projection of $\phi(\circcentre{\splxs})$ into
  $\affhull{\tsplxs}$. Then, letting $R = \circrad{\splxs}$, and
  $\tR = \circrad{\tsplxs}$, and choosing $\tilde{p} = \phi(p) \in
  \tsplxs$, we have
  \begin{equation*}
    \begin{split}
      z^2 &= \distEm{\phi(\circcentre{\splxs})}{\tilde{p}}^2 -
      \distEm{\tilde{p}}{\hat{c}}^2\\
      &\leq (1 + \metlipconst)^2R^2 - (\tR - w)^2\\
      &= R^2 - \tR^2 + 2\tR w + 2R^2\metlipconst + \metlipconst^2 R^2 - w^2.
    \end{split}
  \end{equation*}
  Using \Lemref{lem:circ.rad.distort}, we write $\tR$ in terms of $R$,
  as $\abs{R - \tR} \leq sR$, where
  \begin{equation*}
    s = \frac{16 k^{\frac{3}{2}}
      \metlipconst}{\thickness{\splxs}^3}.
  \end{equation*}
  Then using \Lemref{lem:circ.ball.distort} to bound $w$, and writing
  $\Delta$, and $\Upsilon$, instead of $\longedge{\splxs}$ and
  $\thickness{\splxs}$, we find
  \begin{equation*}
    \begin{split}
      \distEm{\phi(\circcentre{\splxs})}{\circcentre{\tsplxs}}^2
      &\leq R^2 - (1-s)^2 R^2 + 2w(1+s)R  + 2\metlipconst R^2
      + \metlipconst^2 R^2\\
      &\leq  2(sR + w +ws)R + (2 + \metlipconst)\metlipconst R^2\\ 
      &\leq \left[ 2\left( \frac{16 k^{\frac{3}{2}}}{\Upsilon^3}
          + \frac{3 \sqrt{k}R}{\Upsilon \Delta} 
          + \frac{54 k^2 R \metlipconst}{\Upsilon^4 \Delta} \right)
        + \left(2 + \metlipconst \right) \right] R^2 \metlipconst\\
      &\leq \left[ 2\left( \frac{16 k^{\frac{3}{2}}}{\Upsilon^3}
          + \frac{3 \sqrt{k}}{2\Upsilon^2} 
          + \frac{27 k^2  }{16 \Upsilon^3} \right)
        + 3 \right] R^2 \metlipconst
      \qquad \text{using \Lemref{lem:crude.rad.bnd}}\\
      &\leq \left[ \frac{42 k^2  }{ \Upsilon^3} \right] R^2 \metlipconst.
   \end{split}
  \end{equation*}
\end{proof}

    % -*- LaTeX -*-
% more_details.tex
% 20130614 extracted from manmesh.tex (created 20130420)
%
% some details of the proof of validity of the algorithm

\subsection{The size of the domains}
\label{sec:domain.size}

The domains $U_{ij}$ on which the transition functions are defined need to be large enough to accommodate two distinct requirements.  First, the domain of the transition function $\transij$ must contain a large enough neighbourhood of $p'_i$ that we can apply the metric stability result of \cite{boissonnat2013stab1} to ensure that $\str{p'_i; \delof{\ppts_i}}$ will be the same as $\str{p'_i; \delof{\ppts_j} }$ whenever $p'_j \in \str{p'_i; \delof{\ppts_i}}$. The second requirement is that any potential forbidden configuration in the region of interest must lie entirely within the domain of the transition function associated with each of its vertices.

We recall the stability result \cite[Theorem
4.17]{boissonnat2013stab1} that we will use:
\begin{lem}[Delaunay stability under metric perturbation]
  \label{lem:del.metric.stability}
  Suppose $\lptsi$ is a $(\psparseconst, \psamconst_i)$-net and
  $\convhull{\lptsi} \subseteq U \subset \rem$ and $\gdistj:
  U \times U \to \reel$ is such that $\abs{\disti{x}{y} -
    \distj{x}{y}} \leq \metpert$ for all $x,y \in U$. Suppose also
  that $\qpts \subseteq \lptsi$ is a set of interior points such that
  every $m$-simplex $\splxs \in \str{\qpts}$ is $\flakebnd^m$-thick
  and $\delta$-protected and satisfies $\disti{p}{\bdry{U}} \geq
  2\psamconst_i$ for every vertex $p \in \splxs$. If
  \begin{equation*}
    \metpert \leq \frac{\flakebnd^m \psparseconst}{36}\delta,
  \end{equation*}
  then
  \begin{equation*}
    \str{\qpts;\del{\gdistj}{\lptsi}} = \str{\qpts;\del{\gdisti}{\lptsi}}.
  \end{equation*}
\end{lem}
The notation $\del{\gdistj}{\lptsi}$
in \Lemref{lem:del.metric.stability} means that the metric $\gdistj$
is used to compute the Delaunay complex of $\lptsi$. For our
purposes, $\gdistj$ is the pullback by $\transij$ of the Euclidean
metric on $U_j$.  Thus we have the identification
\begin{equation*}
  \str{\qpts ; \del{\gdistj}{\lptsi}} \cong \str{\transij(\qpts) ;
    \delof{\transij(\lptsi)}}. 
\end{equation*}
We will use $\qpts = \{p'_i\}$.  Some argument is required to ensure
that \Lemref{lem:del.metric.stability} provides a route to the desired
equivalence
\begin{equation}
  \label{eq:desired.star.equiv}
  \str{p'_i; \delof{\ppts_i}} \cong \str{p'_i ; \delof{\ppts_j}},
  \quad \text{when } p'_j \in \str{p'_i; \delof{\ppts_i}} .
\end{equation}
We first establish our ``region of interest''.  We demand, for all $i
\in \apts$, that $\pts_i$ be a $(\sparseconst, \samconst_i)$-net for
$\balli{p_i}{8 \samconst_i}$, and we define $\lptsi = \ppts_i \cap
\balli{p_i}{6 \samconst_i}$.  Since $\ppts_i$ changes as the algorithm
progresses, points may come and go from $\lptsi$, but we will ensure
that when the algorithm terminates, $\lptsi$ will contain no forbidden
configurations.
\begin{lem}
  \label{lem:inclusion.star.equivalences}
  For all $i \in \apts$ we have
  \begin{equation*}
    \str{p'_i; \delof{\lptsi}} = \str{p'_i ; \delof{\ppts_i}}.
  \end{equation*}
  and if $p \in \str{p'_i; \delof{\ppts_i}}$, then
  $\disti{p}{\bdry{\balli{p_i}{6\samconst_i}}} > 2\psamconst_i$.

  If $\balli{p_i}{6 \samconst_i} \subseteq U_{ij}$ whenever $p'_j
  \in \str{p'_i; \delof{\ppts_i}}$, then
 \begin{equation*}
   \str{p'_i ; \delof{ \transij(\lptsi)}} = \str{p'_i ; \delof{\ppts_j}}.
 \end{equation*}
\end{lem}
\begin{proof}
  We have $\disti{p'_i}{\bdry{\balli{p_i}{6\samconst_i}}} \geq
  \frac{24}{5}\psamconst_i - \frac{1}{4}\psamconst_i > 4
  \psamconst_i$. The density assumption guarantees that if $\splxs^m
  \in \str{p'_i ; \delof{\lptsi}}$, then $\circrad{\splxs^m} <
  \psamconst_i$, and the observation that
  $\balli{\circcentre{\splxs^m}}{\circrad{\splxs^m}} \subset
  \balli{p_i}{4\samconst_i}$, leads to the first equality, and the
  bound on the distance from $p$ to
  $\bdry{\balli{p_i}{6\samconst_i}}$.

  The second equality follows from two observations. First we show
  that if $\splxs^m \in \str{p'_i; \delof{\transij(\lptsi)}}$, then
  $\circrad{\splxs} < \psamconst_j$. Since $\transij(\lptsi) \subset
  \ppts_j$, and $\ppts_j$ is  $\psamconst_j$-dense for
  $B = \ballj{p_j}{8\samconst_j}$, it is sufficient to show that
  $\distj{p'_i}{\bdry{B}} \geq 2\psamconst_j$. Since $p'_j \in
  \str{p'_i; \delof{\ppts_i}}$, we have $\distj{p'_i}{p'_j} \leq (1 +
  \metlipconst)\disti{p'_i}{p'_j} \leq (1 + \metlipconst)2
  \psamconst_i \leq 2(1+ \metlipconst)(1 +
  \lipsamconst)\frac{5}{4}\samconst_j \leq 5\samconst_j$. Thus since
  $\distj{p_j}{p'_j} \leq \frac{1}{4}\samconst_j$, we have
  $\distj{p'_j}{\bdry{B}} \geq 8 \samconst_j - \frac{21}{4}\samconst_j
  = \frac{11}{4}\samconst_j \geq \frac{11}{5}\psamconst_j$.
  This establishes that the Delaunay ball for $\splxs^m$ must remain
  empty when points outside of $B$ are considered.

  The second obervation required to establish the second equality is
  that if $q \in \ppts_j$ is such that $\distj{p'_i}{q} <
  \psamconst_j$, then $q \in
  \transij(\balli{p_i}{6\samconst_i})$. Indeed, we have
  $\disti{p'_i}{q} \leq 2(1 + \metlipconst)\psamconst_j \leq
  2(1+\metlipconst)(1 + \lipsamconst)\frac{5}{4}\samconst_i \leq
  5\samconst_i$. The result follows since $\disti{p_i}{p'_i} \leq
  \frac{1}{4}\samconst_i$. 
\end{proof}

If $p'_j \str{p'_i ; \delof{\ppts_i}}$, then $\disti{p_i}{p_j} <
\frac{1}{2}\samconst_i + 2\psamconst_i \leq 3\samconst_i$. Thus
\Lemref{lem:inclusion.star.equivalences} establishes the first
requirement on $U_{ij}$, namely
\begin{equation}
  \label{eq:domain.req.metric}
  \balli{p_i}{6\samconst_i} \subset U_{ij} \quad \text{if }
  \disti{p_i}{p_j} < 3 \samconst_i.
\end{equation}
The second requirement arises from the fact that we wish to ensure
that there are no forbidden configurations in $\lptsi$. This will be
sufficient for us to apply \Lemref{lem:del.metric.stability}.
\begin{lem}[Protected stars]
  \label{lem:protected.stars.appdx}
  \label{lem:protected.stars}
  If there are no forbidden configurations in $\lptsi$, then all the
  $m$-simplices in $\str{p'_i; \delof{\lptsi} }$ are $\flakebnd$-good
  and $\delta$-protected, with $\delta = \delta_0 \psparseconst
  \psamconst_i$.
\end{lem}
\begin{proof}
  Since $\ppts_i$ is a $(\psparseconst, \psamconst_i)$-net for
  $\balli{p_i}{8\samconst_i}$, it follows that $\lptsi$ is a
  $(\psparseconst, \psamconst_i)$-net. 
  % Indeed, if $x \in
  % \convhull{\lptsi}$ satisfies $\disti{x}{\bdry{\convhull{\lptsi}}}
  % \geq \psamconst_i$, then there is a $q \in \ppts_i$ with
  % $\disti{x}{q} < \psamconst_i$, and so $q \in \lptsi$ also.
  Thus if there are no forbidden configurations in $\lptsi$, then by
  \cite[Lemma 3.6]{boissonnat2014flatpert}, all the
  $m$-simplices in $\rdelsmhull{\lptsi}$ will be $\flakebnd$-good and
  $\delta$-protected, with $\delta = \delta_0 \psparseconst
  \psamconst_i$.

  The sampling criteria ensure that every point on
  $\bdry{\convhull{\lptsi}}$ must be at a distance of less than
  $2\psamconst_i$ from $\bdry{\balli{p_i}{6\samconst_i}}$. Thus
  $\disti{p_i}{\bdry{\convhull{\lptsi}}} > 6\samconst_i -
  2\psamconst_i \geq \frac{14}{5}\psamconst_i$.  Also,
  $\disti{p_i}{p'_i} \leq \frac{\psamconst_i}{4}$, and we find
  that $\disti{p'_i}{\bdry{\convhull{\lptsi}}} \geq \frac{51}{20}
  \psamconst_i$. Thus, since $\disti{p'_i}{\circcentre{\splxs}} <
  \psamconst_i$ if $\splxs$ is in $\str{p'_i ; \delof{\lptsi} }$, we
  have $\str{p'_i ; \delof{\lptsi}} \subseteq \rdelsmhull{\lptsi}$,
  and hence the result.
\end{proof}

According to \Lemref{lem:prop.forbid.cfg}~\ref{hyp:diam.bnd}, if
$\splxt$ is a forbidden configuration in $\lptsi$, then
$\longedge{\splxt} < \frac{15}{4} \samconst_i$, and it follows that if
$p'_j \in \splxt$, then $\splxt \subset \balli{p_j}{4
  \samconst_i}$. We will require that each potential forbidden
configuration in $\lptsi$ lies within the domain of any transition
function associated one of its vertices. Thus we demand that
\begin{equation}
  \label{eq:domain.req.forbid}
  \balli{p_j}{4 \samconst_i} \cap \balli{p_i}{6\samconst_i} \subset
  U_{ij} \quad \text{if } p_j \in \balli{p_i}{6\samconst_i}.
\end{equation}
For simplicity we accommodate Equations~\eqref{eq:domain.req.metric}
and \eqref{eq:domain.req.forbid} by demanding that
\begin{equation}
  \label{eq:trans.domain.req}
  \balli{p_j}{9 \samconst_i} \cap \balli{p_i}{6\samconst_i} \subset
  U_{ij} \quad \text{if } p_j \in \balli{p_i}{6\samconst_i}.
\end{equation}

In summary, Lemmas \ref{lem:del.metric.stability},
\ref{lem:inclusion.star.equivalences}, and \ref{lem:protected.stars.appdx}
combine to yield the desired equivalence of
stars~\eqref{eq:desired.star.equiv}, under the assumption that $\lptsi$
has no forbidden configurations.  We take $U =
\balli{p_i}{6\samconst_i}$ in \Lemref{lem:del.metric.stability}, and
\Eqnref{eq:metric.distortion} yields $\metpert \leq \metlipconst 12
\samconst_i$. Using $\delta = \flakebnd^{m+1}\psparseconst
\psamconst_i \geq \frac{1}{2}\flakebnd^{m+1}\sparseconst \samconst_i$
we have:
\begin{lem}[Stable stars]
  \label{lem:stable.stars.appdx}
  \label{lem:stable.stars}
  If
  \begin{equation*}
    \metlipconst \leq \frac{\flakebnd^{2m+1}\sparseconst^2}{2^{12}},
  \end{equation*}
  and there are no forbidden configurations in $\lptsi$, then for all
  $p'_j \in \str{p'_i; \delof{\ppts_i}}$, we have
  \begin{equation*}
    \str{p'_i; \delof{\ppts_i}} \cong \str{p'_i; \delof{\ppts_j}}.
  \end{equation*}
\end{lem}

We have established minimal requirements on the size of the domains
$U_{ij}$, but these requirements may implicitly demand more. Although
$\transij: U_{ij} \to U_{ji}$ is close to an isometry, $\samconst_i$
may be almost twice as large as $\samconst_j$. Thus the requirement on
$U_{ij}$ may imply that $U_{ji} = \transij(U_{ij})$ is significantly
larger than \Eqnref{eq:trans.domain.req} demands.

Clearly we must have
\begin{equation*}
  \bigcup_{j \in \nbrsi} U_{ij} \subset U_i.
\end{equation*}
We have also explicitly demanded that $\pts_i$ be a $(\sparseconst,
\samconst_i)$-net for $\balli{p_i}{8\samconst_i}$. We will assume that
$\balli{p_i}{8\samconst_i} \subset U_i$.

\subsection{Hoop distortion}
\label{sec:hoop.distort}

We will rely primarily on Properties~\ref{hyp:clean.hoop.bnd} and
\ref{hyp:good.facets} of forbidden configurations (
\Lemref{lem:prop.forbid.cfg}), and the stability of the circumcentres
exhibited by \Lemref{lem:bound.cc}. We have the following observation
about the properties of forbidden configurations under the influence
of the transition functions:
\begin{lem}
  \label{lem:raw.hoop.distort}
  Assume $\metlipconst \leq \left( \frac{\flakebnd^k}{4} \right)^2$.
  If $\splxt = \splxjoin{p'_i}{\splxs} \subset \lptsj \subset U_j$ is
  a forbidden configuration, where $\splxs$ is a $k$-simplex,
  then
  $\tsplxs = \trans{j}{i}(\splxs) \subset \ppts_i$ is
  $\tflakebnd^k$-thick, with
  \begin{equation*}
    \tflakebnd^k = \frac{2}{5 \sqrt{k}} \flakebnd^k,
  \end{equation*}
  has a radius satisfying
  \begin{equation*}
%    \label{eq:def.rad.bnd}
    \circrad{\tsplxs} \leq  2 \left( 1 + \frac{16 k^{\frac{3}{2}}
        \metlipconst}{\flakebnd^{3k} } \right)(1 + \lipsamconst)\samconst_i,
  \end{equation*}
  and $\disti{p'_i}{\diasphere{\tsplxs}} \leq 2\thoopbnd\samconst_i$,
  where
  \begin{equation*}
    \thoopbnd = 
    \left(
          \hoopbnd(1 +  \metlipconst) +
          \left(
            \frac{ 12 k^{\frac{3}{2}} }{ \flakebnd^{2k} }
          \right)\metlipconst^{\frac{1}{2}} 
        \right)(1 + \lipsamconst).
 \end{equation*}
\end{lem}
\begin{proof}
  The bound for $\tflakebnd^k$ follows immediately from
  \Lemref{lem:intrinsic.thick.distortion}, and the fact that $\splxs$
  is $\flakebnd^k$-thick
  (\Lemref{lem:prop.forbid.cfg}~\ref{hyp:good.facets}). Likewise, the
  radius bound is a direct consequence of
  \Lemref{lem:circ.rad.distort} and
  \Lemref{lem:prop.forbid.cfg}~\ref{hyp:clean.facet.rad.bnd}.

  The bound on $\thoopbnd$ is obtained from
  Property~\ref{hyp:clean.hoop.bnd} with the aid of Lemmas
  \ref{lem:circ.rad.distort} and \ref{lem:bound.cc}.  We have
  $\disti{p'_i}{\diasphere{\tsplxs}} =
  \abs{\disti{p'_i}{\circcentre{\tsplxs}} - \circrad{\tsplxs}}$, and
  we are able to get a tighter upper bound on $\circrad{\tsplxs} -
  \disti{p'_i}{\circcentre{\tsplxs}}$, than we can for
  $\disti{p'_i}{\circcentre{\tsplxs}} - \circrad{\tsplxs}$. Thus
  \begin{equation*}
    \begin{split}
      \disti{p'_i}{\diasphere{\tsplxs}}
      &=
      \abs{\disti{p'_i}{\circcentre{\tsplxs}} - \circrad{\tsplxs}}\\
      &\leq
      (1 + \metlipconst)\distj{p'_i}{\circcentre{\splxs}} +
      \disti{\trans{j}{i}(\circcentre{\splxs})}{\circcentre{\tsplxs}}
      - (\circrad{\splxs} - \abs{ \circrad{\tsplxs} - \circrad{\splxs} })\\
      &\leq
      (1 + \metlipconst)(\hoopbnd \circrad{\splxs} + \circrad{\splxs}) 
      - \circrad{\splxs}
      + \disti{\trans{j}{i}(\circcentre{\splxs})}{\circcentre{\tsplxs}}
      + \abs{ \circrad{\tsplxs} - \circrad{\splxs} }\\
      &\leq
      \left( \hoopbnd (1 + \metlipconst) 
        +
        \metlipconst 
        +
        \left[ \left( \frac{42 k^2}{\thickness{\splxs}^3} \right)  \metlipconst
        \right]^{\frac{1}{2}} 
        +
        \frac{16 k^{\frac{3}{2}} \metlipconst}{\thickness{\splxs}^3 }
      \right) \circrad{\splxs}\\
      &\leq
      2\left(
        \hoopbnd(1 +  \metlipconst)
        +
        \metlipconst
        +
        \left[ \left( \frac{42 k^2}{\flakebnd^{3k}} \right)  \metlipconst
        \right]^{\frac{1}{2}} 
        +
        \frac{16 k^{\frac{3}{2}} \metlipconst}{ \flakebnd^{3k}}
      \right) \samconst_j\\ 
        &\leq
        2\left(
          \hoopbnd(1 +  \metlipconst) +
          \left(
            \metlipconst^{\frac{1}{2}}
            +
            \frac{7k}{ \flakebnd^{\frac{3k}{2}} } 
            +
            \frac{16 k^{\frac{3}{2}} \metlipconst^{\frac{1}{2}}}{ \flakebnd^{3k} }
          \right)\metlipconst^{\frac{1}{2}}
        \right)(1 + \lipsamconst) \samconst_i\\
        &\leq
        2\left(
          \hoopbnd(1 +  \metlipconst) +
          \left(
            \frac{ 12 k^{\frac{3}{2}} }{ \flakebnd^{2k} }
          \right)\metlipconst^{\frac{1}{2}} 
        \right)(1 + \lipsamconst) \samconst_i.
   \end{split}
  \end{equation*}
\end{proof}
We have abused the notation slightly because $\tsplxt =
\trans{j}{i}(\splxt)$ need not actually satisfy the
$\thoopbnd$-hoop property definition
$\disti{p}{\circsphere{\opface{p}{\tsplxt}}} \leq \thoopbnd
\circrad{\opface{p}{\tsplxt}}$, because $\circrad{\tsplxt}$ may be
larger than $2\samconst_i$. However we are not concerned with the
$\thoopbnd$-property for $\tsplxt$; instead we desire a condition that
will permit the extended algorithm to emulate the original Euclidean
perturbation algorithm \cite{boissonnat2014flatpert}, and
 guarantee that forbidden configurations such as $\splxt$ cannot exist
 in any of the sets $\lptsj$.

The bounds in \Lemref{lem:raw.hoop.distort} can be further
simplified. We have announced them in this intermediate state in order
to elucidate the roles played by $\metlipconst$ and $\lipsamconst$.
In particular, there is no need to significantly constrain
$\lipsamconst$. The original perturbation algorithm for points in
Euclidean space \cite{boissonnat2014flatpert} extends to the case
of a non-constant sampling radius simply by replacing $\hoopbnd$ by
$\thoopbnd \leq (1 + \lipsamconst)\hoopbnd \leq 2 \hoopbnd$, as can be
seen by setting $\metlipconst = 0$ in the expression for $\thoopbnd$.

In the general case of interest here, we see from the espression for
$\thoopbnd$ presented in \Lemref{lem:raw.hoop.distort}, that
$\metlipconst$ must be considerably more constrained with respect to
$\flakebnd$ if we are to obtain an expression for $\thoopbnd$ that
goes to zero as $\flakebnd$ goes to zero. For the purposes of the
algorithm, we do not require the bounds on the radius or the
thickness. 
\begin{lem}[Hoop distortion]
  \label{lem:hoop.distort.appdx}
  \label{lem:hoop.distort}
  If
  \begin{equation*}
    \metlipconst \leq \left( \frac{\flakebnd^{2m + 1}}{4} \right)^2,
  \end{equation*}
  then for any forbidden configuration $\splxt =
  \splxjoin{p'_j}{\splxs} \subset \lptsi$, there is a simplex $\tsplxs
  = \transij(\splxs) \subset \ppts_j$ such that 
  $\distj{p'_j}{\diasphere{\tsplxs}} \leq 2\thoopbnd\samconst_j$,
  where
  \begin{equation*}
    \thoopbnd =  \frac{2^{16} m^{\frac{3}{2}} \flakebnd}{\sparseconst^3}.
  \end{equation*}
\end{lem}
\begin{proof}
  By the properties of a forbidden configuration, $\splxs$ is a
  $k$-simplex with $k \leq m$.
  From \Lemref{lem:raw.hoop.distort},
  \begin{equation*}
    \begin{split}
      \thoopbnd
      &= 
      \left(
        \hoopbnd(1 +  \metlipconst) +
        \left(
          \frac{ 12 k^{\frac{3}{2}} }{ \flakebnd^{2k} }
        \right)\metlipconst^{\frac{1}{2}} 
      \right)(1 + \lipsamconst)\\
      &\leq
      2\left(
        \frac{2^{13}\flakebnd}{\sparseconst^3}
        (1 +  \metlipconst)
       + \left(
          \frac{ 12 m^{\frac{3}{2}} }{ \flakebnd^{2m} } \right)
        \frac{\flakebnd^{2m + 1}}{4} 
      \right)\\
      &<
      \left(
        ( 1 +  2^{-4} )
        +  m^{\frac{3}{2}}
      \right) \frac{2^{14}\flakebnd}{\sparseconst^3}\\
      &< \frac{2^{16} m^{\frac{3}{2}} \flakebnd}{\sparseconst^3}.
   \end{split}
  \end{equation*}
\end{proof}
% The perturbation algorithm will ensure that such simplices cannot
% exist.

%
%\clearpage
%\appendix
%\pagenumbering{Roman}
%
% \input{splx_distort}
%     \input{more_details}
%\input{quality}
%\input{homeomorphism}
%\input{munkres}
%    \input{outmetric}

\subsection*{Acknowledgements}

We gratefully profited from discussions with Mathijs Wintraeken.  

%This
%research has been partially supported by the 7th Framework Programme
%for Research of the European Commission, under FET-Open grant number
%255827 (CGL Computational Geometry Learning). 

% Normally AFTER appendix
\phantomsection
\bibliographystyle{alpha}
\addcontentsline{toc}{section}{Bibliography}
\bibliography{delrefs}

\end{document}